\documentclass[12pt,a4paper]{elsarticle}
\usepackage[T1]{fontenc}
\usepackage{rotating}

\usepackage[dvipsnames,table]{xcolor}

\usepackage{amsmath, amscd, amsthm, amsfonts, amssymb, graphicx, subcaption, color, soul, enumerate, xcolor,  mathrsfs, latexsym, bigstrut, framed, caption}
\usepackage{multirow}
\usepackage{nicematrix}
\usepackage{pstricks,pst-node,pst-tree}

\usepackage[colorlinks, plainpages, backref, pdftoolbar=true]{hyperref}
\AtBeginDocument{%
\hypersetup{linkcolor=darkblue, anchorcolor=green, citecolor=midblue, urlcolor=darkblue, filecolor=magenta}
}

\usepackage{algorithm}
\usepackage{algpseudocode}

\usepackage{etoolbox}

\usepackage{setspace}
\usepackage{tikz}
\usetikzlibrary{calc,decorations.pathreplacing,decorations.markings,positioning,shapes,math,backgrounds}

\pgfdeclarelayer{background}
\pgfdeclarelayer{foreground}
\pgfsetlayers{background,main,foreground}

\usepackage[normalem]{ulem}

\newcommand{\Prob}{\ensuremath{\operatorname{Pr}}}

\newcommand{\w}{\color{black}}

\definecolor{vibtilg}{HTML}{00FF99}
\definecolor{C_greedy}{HTML}{1F77B4}
\definecolor{C_NGC}{HTML}{FF7F0E}
\definecolor{C_growth}{HTML}{2CA02C}
\definecolor{C_LMC}{HTML}{D62728}
\definecolor{C_LS}{HTML}{9467BD}
\definecolor{C_ELS}{HTML}{8C564b}
\definecolor{C_PLS}{HTML}{E377C2}
\definecolor{C_rnd}{HTML}{0033CC}
\definecolor{des}{HTML}{edeb77}

\definecolor{cupgreen}{rgb}{0,0.498,0.208}
\definecolor{cupblue}{rgb}{0,0,.5}
\definecolor{cupred}{rgb}{1,0.04,0}
\definecolor{cuppink}{rgb}{0.925,0,0.545}
\definecolor{cupmagenta}{rgb}{0.624,0.161,0.424}
\definecolor{cupbrown}{rgb}{0.71,0.212,0.133}

\definecolor{TITLE}{rgb}{0,0,0}
\definecolor{midblue}{rgb}{0.00,0.0,0.80}
\definecolor{darkblue}{rgb}{0.00,0.00,0.45}
\definecolor{SECTION}{rgb}{0.50,0.00,1.00}
\definecolor{THM}{rgb}{0.8,0,0.1}
\definecolor{SEC}{rgb}{0,0,1}
\definecolor{calpolypomonagreen}{rgb}{0.12, 0.3, 0.17}
\definecolor{byzantine}{rgb}{0.64, 0.14, 0.54}
\definecolor{cadmiumgreen}{rgb}{0.0, 0.42, 0.24}
\definecolor{caribbeangreen}{rgb}{0.0, 0.6, 0.46}

\makeatother
\normalsize

\newtheorem{theorem}{{\color{THM} Theorem}}[section]

\theoremstyle{definition}

\numberwithin{equation}{section}
\newtheorem*{definition*}{\color{THM}Definition}
\newtheorem*{observation*}{\color{THM}Observation}

\def\NoNumber#1{{\def\alglinenumber##1{}\State #1}\addtocounter{ALG@line}{-1}}
\journal{Computers \& Operations Research}

\begin{document}
\begin{frontmatter}
    \title{Local Search Improvements for Soft Happy Colouring}
    
    \author[inst1]{Mohammad H. Shekarriz}
    \ead{m.shekarriz@deakin.edu.au}
    \affiliation[inst1]{
        organization={School of Information Technology, Deakin University},
        city={Geelong},
        state={VIC},
        country={Australia}}
    \author[inst1]{Dhananjay Thiruvady}
    \ead{dhananjay.thiruvady@deakin.edu.au}
    \author[inst1]{Asef Nazari}
    \ead{asef.nazari@deakin.edu.au}
    \author[inst2]{Wilfried Imrich}
    \ead{imrich@unileoben.ac.at}
    \affiliation[inst2]{
        organization={Montanuniversit\"at Leoben},
        city={Leoben},
        state={},
        postcode={A-8700},
        country={Austria}
    }
    \begin{abstract}
        For $0\leq \rho\leq 1$ and a coloured graph $G$, a vertex $v$ is $\rho$-happy if at least $\rho \deg(v)$ of its neighbours have the same colour as $v$. Soft happy colouring of a partially coloured graph $G$ is the problem of finding a vertex colouring $\sigma$ that preserves the precolouring and has the maximum number of $\rho$-happy vertices. It is already known that this problem is NP-hard and directly relates to the community structure of the graphs; under certain condition on the proportion of happiness $\rho$ and for graphs with community structures, the induced colouring by communities can make all the vertices $\rho$-happy. We show that when $0\leq \rho_1<\rho_2\leq 1$, a complete $\rho_2$-happy colouring has a higher accuracy of community detection than a complete $\rho_1$-happy colouring. Moreover, when $\rho$ is greater than a threshold, it is unlikely for an algorithm to find a complete $\rho$-happy colouring with colour classes of almost equal sizes. Three local search algorithms for soft happy colouring are proposed, and their performances are compared with one another and other known algorithms. Among them, the linear-time local search is shown to be not only very fast, but also a reliable algorithm that can dramatically improve the number of $\rho$-happy vertices. 
    \end{abstract}
    
\end{frontmatter}

\section{Introduction}\label{Sec:Intro}
Vertex colouring partitions a graph's vertex set into disjoint ``colour'' sets, and solutions to many complex problems are presented in this form. 
One type that has recently attracted considerable interest is \emph{happy colouring}~\cite{ZHANG2015117}. This problem was introduced in 2015 by Zhang and Li as a vertex colouring with a maximum number of \emph{happy vertices}, those that have the same colour as their neighbours. One generalisation of this colouring is \emph{soft happy colouring}, which employs a \emph{proportion of happiness}, $0\leq \rho \leq 1$, and considers maximising the number of \emph{$\rho$-happy} vertices, those that at least $\rho \deg(v)$ of their neighbours have the same colour as them~\cite{ZHANG2015117,SHEKARRIZ2025106893}.

Finding a happy colouring, in a connected graph with some vertices precoloured, is a difficult problem and is mostly concerned with \emph{maximising the number of happy vertices in a partially coloured graph}. Several approaches have been developed for finding happy colourings. For example, in~\cite{Lewis2019265} integer programming based approaches {\w are} devised for happy colouring, in~\cite{Zhang2018}~ a randomized LP-rounding technique and a non-uniform approach are developed, in~\cite{thiruvady2020} and \cite{LEWIS2021105114} tabu search approaches are proposed, and in~\cite{THIRUVADY2022101188} evolutionary algorithms and hybrids of metaheuristics and matheuristics are investigated. 

The idea behind introducing happy colouring and its generalisations was \emph{homophily} in social networks~\cite{Homophily}, a concept that the \emph{community structure} of graphs can fully express. A \emph{community}~\cite{10.1007/978-3-540-48413-4_23} in a graph is a group of vertices of ``sufficient size'' whose adjacencies among themselves are denser than edges that connect them to other groups. Figure~\ref{fig:communities} illustrates this idea for a graph with 4 communities.

\begin{figure}
\centering
\begin{tikzpicture}[scale=0.8]

\tikzstyle{comm1} = [draw=green!20, fill=green!10, thick, smooth]
\tikzstyle{comm2} = [draw=orange!20, fill=orange!10, thick, smooth]
\tikzstyle{comm3} = [draw=magenta!20, fill=magenta!10, thick, smooth]
\tikzstyle{comm4} = [draw=red!20, fill=red!10, thick, smooth]

\tikzstyle{vertex} = [circle, draw, fill=white, minimum size=8pt, inner sep=0pt]
\tikzstyle{comm1v} = [circle, draw, fill=green!50, minimum size=8pt, inner sep=0pt]
\tikzstyle{comm2v} = [circle, draw, fill=orange!50, minimum size=8pt, inner sep=0pt]
\tikzstyle{comm3v} = [circle, draw, fill=magenta!50, minimum size=8pt, inner sep=0pt]
\tikzstyle{comm4v} = [circle, draw, fill=red!50, minimum size=8pt, inner sep=0pt]

\fill[comm1] (-5.5,0.5) .. controls (-4,3.2) and (-3,1.2) .. (-3.3,-0.5) .. controls (-4.5,-3) and (-6.3,-1) .. cycle;
\fill[comm2] (2.8,0.5) .. controls (4.5,3) and (5.9,1.2) .. (5.7,-0.5) .. controls (4.5,-2) and (3,-1) .. cycle;
\fill[comm3] (-1.1,2.3) .. controls (0.5,0.8) and (2,3.5) .. (1,4.5) .. controls (-1,4.8) and (-2,3.5) .. cycle;
\fill[comm4] (-0.65,-2.3) .. controls (0.6,-1.4) and (2.5,-1.5) .. (1.2,-4.3) .. controls (-0.5,-4.9) and (-1.5,-3.8) .. cycle;

\node[comm1v] (v1) at (-5.2,0.5) {};
\node[comm1v] (v2) at (-4,1.6) {};
\node[comm1v] (v3) at (-3.5,0.5) {};
\node[comm1v] (v4) at (-3.5,-0.5) {};
\node[comm1v] (v5) at (-4.5,-1.5) {};
\node[comm1v] (v6) at (-5.5,-0.5) {};

\node[comm2v] (v7) at (3,0.5) {};
\node[comm2v] (v8) at (4.5,1.5) {};
\node[comm2v] (v9) at (5.5,0.5) {};
\node[comm2v] (v10) at (5.5,-0.5) {};
\node[comm2v] (v11) at (4.3,-1) {};
\node[comm2v] (v12) at (4.5,0.5) {};

\node[comm3v] (v13) at (-1,2.5) {};
\node[comm3v] (v14) at (0,2.1) {};
\node[comm3v] (v15) at (1,3) {};
\node[comm3v] (v16) at (0,4.3) {};
\node[comm3v] (v17) at (-1,3.9) {};
\node[comm3v] (v18) at (0.9,4.3) {};

\node[comm4v] (v19) at (-0.5,-2.5) {};
\node[comm4v] (v20) at (0.3,-3) {};
\node[comm4v] (v21) at (1.46,-3) {};
\node[comm4v] (v22) at (0.8,-2) {};
\node[comm4v] (v23) at (-0.5,-4) {};
\node[comm4v] (v24) at (1.1,-4.14) {};

\node[comm1v] (v25) at (-4.2,0.5) {};

\draw (v1) -- (v2) -- (v3) -- (v4) -- (v1) -- (v5) -- (v3);
\draw (v5) --(v25)--(v2) -- (v6) -- (v5);

\draw (v7) -- (v8) -- (v9) -- (v10) -- (v7) -- (v11) -- (v9);
\draw (v8) -- (v12);

\draw (v13) -- (v14) -- (v15) -- (v16) -- (v13) -- (v17) -- (v15) -- (v18);
\draw (v14) -- (v18) --(v16) -- (v17);

\draw (v19) -- (v20) -- (v21) -- (v22) -- (v19) -- (v23) -- (v21);
\draw (v22) -- (v20) -- (v24) -- (v21);

\draw [color=black, dashed] (v4) -- (v7);
\draw [color=black, dashed] (v2) -- (v17);
\draw [color=black, dashed] (v8) -- (v15);
\draw [color=black, dashed] (v13) -- (v22);
\draw [color=black, dashed] (v14) -- (v19);
\draw [color=black, dashed] (v5) -- (v23);
\draw [color=black, dashed] (v11) -- (v21);

\end{tikzpicture}

\caption{A connected graph with 4 communities. Dashed lines are inter-community edges. The presented colouring is $\frac{2}{3}$-happy.}\label{fig:communities}
\end{figure}
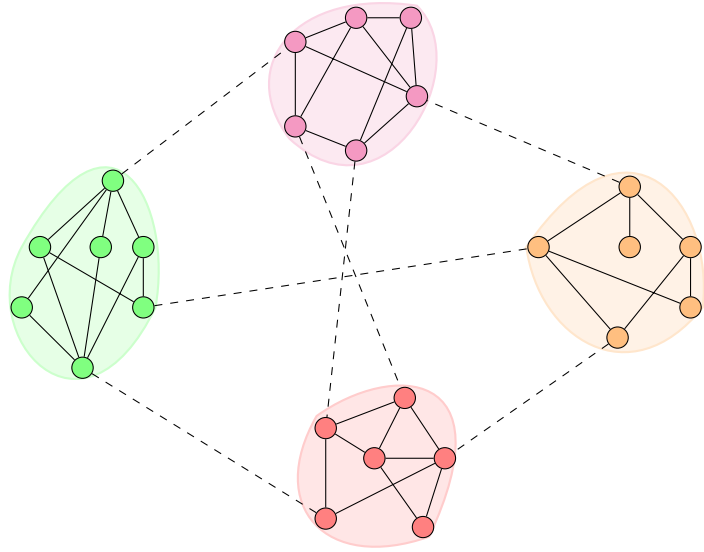

Although real-world graphs almost always have community structures~\cite{Cherifi2019, doi:10.1073/pnas.122653799, PhysRevE.68.065103}, for connected graphs, conventional happy colouring {\w methods} never fully achieve {\w a} partitioning {\w of} the vertex set into communities. One reason is that any colouring with more than one colour has some unhappy vertices in connected graphs. Lewis et al.~\cite{Lewis2019265} pointed out this fact and showed that the situation is worse in dense and large graphs because their proposed integer program finds very few happy vertices in these graphs\footnote{The findings of~\cite{SHEKARRIZ2025106893} and the results of this paper also affirm that there is almost no $1$-happy vertex in dense or large graphs.}. Therefore, relaxing the condition of ``all neighbours have the same colour as $\sigma(v)$'' should be helpful to get closer to the community structure. Zhang and Li~\cite{ZHANG2015117} in their original paper first introduced such relaxations, thereby generalising happy colouring to soft happy colouring, though soft happy colourings (and related variants) have been neglected for over a decade\footnote{It must also be noted that conventional happy colouring is a special case of soft happy colouring (when $\rho=1$).}.

Suppose that some vertices of a graph $G$ are precoloured using $k$ colours ($k\geq 2$), and $\sigma$ is a $k$-colouring extending the precolouring to all the vertices. Then a vertex $v$ is \emph{$\rho$-happy} if at least $\lceil \rho   \deg(v) \rceil$ of neighbours of $v$ have the same colour as $\sigma(v)$. A vertex colouring $\sigma$ is called a soft happy colouring (with $k$ colours) for $G$ if it has the maximum number of $\rho$-happy vertices among such $k$-colouring extensions of the precolouring of $G$. Shekarriz et al. showed that soft happy colouring is related to the community structure of graphs, and if $\rho$ is small enough, all the vertices in a dense graph with communities can be $\rho$-happy~\cite{SHEKARRIZ2025106893}. They also showed that this almost certainly happens as the number of vertices tends to infinity. Their argument is probabilistic and is based on the \emph{Stochastic Block Model} (SBM)~\cite{HOLLAND1983109}, a random graph model for graphs with community structure.

In this paper, we prove that for a graph $G$ in the SBM, a complete $\rho_2$-happy $k$-colouring has an accuracy of community detection higher than a complete $\rho_1$-happy $k$-colouring, provided that $0\leq\rho_1<\rho_2\leq 1$ and the colour classes are of almost the same size. Moreover, we show that if $\rho$ is bounded below by a threshold $\tilde{\xi}$ and {\w if} the SBM graph $G$ has ample of vertices, then the probability of finding a complete $\rho$-happy colouring is 0, while the probability of $\rho$-happiness of an arbitrary vertex approaches 0 when the number of vertices goes to infinity. These results are tested for a large set of randomly generated graphs in the SBM.

Zhang and Li~\cite{ZHANG2015117} introduced two heuristic algorithms to find ``near'' soft happy colouring of a graph, namely {\sf Greedy-SoftMHV} and {\sf Growth-SoftMHV}. Shekarriz et al.~\cite{SHEKARRIZ2025106893} introduced two more heuristics, namely {\sf LMC} and {\sf NGC}. To our knowledge, these are all the {\w currently} known algorithms to tackle soft happy colouring problems.  This paper presents three {\w new} local search algorithms, namely {\sf Local Search} ({\sf LS}), {\sf Repeated Local Search} ({\sf RLS}) and {\sf Enhanced Local Search} ({\sf ELS}) to be used as heuristics or to improve existing solutions. We tested these algorithms for a large set of randomly generated graphs in the SBM to compare their running results with the four previously known algorithms.

The paper is organised as follows: in Section~\ref{sec:th-back}, we present the theoretical background. This includes notations from graph theory, the SBM, and soft happy colouring for graphs in the SBM. In Section~\ref{sec:theory}, we present and prove our theoretical results. Known heuristics are visited in Section~\ref{sec:known} and our local search improvement algorithms are introduced in Section~\ref{sec:LS}. The experimental analysis, validating our theoretical results and comparing the performance of the algorithms, is presented in Section~\ref{sec:analys}. Section~\ref{sec:conc} concludes the paper. 

\section{Theoretical Background}\label{sec:th-back}

Throughout the paper, we use notation and definitions from graph theory.
{\w By a graph $G$, we mean a finite simple graph. It consists of a vertex set $V(G)$ and of a set of edges $E(G)$, whose elements are unordered pairs of distint vertices.} We write $G=(V(G), E(G))$. The numbers of vertices and edges are usually denoted by $n$ and $m$. The set of neighbours of a vertex $v$ is denoted by $N(v)$. We write $u\sim v$ when two vertices $u$ and $v$ are adjacent, and  $u\not\sim v$ when they are not.

This study is concerned with random graphs~\cite{Bollobas_2001}, more specifically, graphs in the Stochastic Block Model (SBM)~\cite{HOLLAND1983109,JERRUM1998155}. The model we use is the simplified version of the SBM, denoted by $\mathcal{G}(n,k,p,q)$, which is the probability space consisting of all graphs on $n$ vertices, with an assignment to $k$ vertex-disjoint communities. The probability of having an edge between two vertices of the same community is $p$, while two vertices of different communities are adjacent with the probability of $q$. We always suppose that $q<p$ to have a meaningful community structure. Moreover, to ease understanding of theoretical results and validate practical ones, we assume that each community is almost equal in size to $\frac{n}{k}$. For more information about the SBM, see~\cite{Lee2019}.

For $k\geq 2$, the problem of soft happy colouring is to extend a partial $k$-colouring to a complete vertex $k$-colouring so that the number of $\rho$-happy vertices is maximum among all such extensions of the partial colouring. The uncoloured vertices in the original partial colouring are called \emph{free vertices}. The number of $\rho$-happy vertices of a colouring $\sigma$ of the graph $G$ is denoted by $H_\rho (G,\sigma)$ or simply by $H_\rho (\sigma)$. The \emph{ratio of $\rho$-happy vertices} in a colouring $\sigma$ is $\alpha(\sigma)=\frac{H_\rho (\sigma)}{n}$. If $\sigma$ is a complete $\rho$-happy colouring, i.e. $\alpha(\sigma)=1$, we write $\sigma\in H_\rho$, and if $G$'s communities induce a $\rho$-happy colouring, we write $G\in H_\rho$. 

Note that in our soft happy colouring problem instances, we always have some precoloured vertices; at least one vertex from each community is assigned the colour of its community number. Therefore, the number of vertices that are assigned colours in accordance with their original community is nonzero. \emph{Accuracy of community detection} of a colouring $\sigma$ is the ratio of vertices whose colours agree with their original communities, and it is denoted by $ACD(\sigma)$.

Shekarriz et al in~\cite{SHEKARRIZ2025106893} investigated soft happy colouring for graphs in the SBM. Assuming that $G\in\mathcal{G}(n,k,p,q)$, $n$ is large enough, $k \ge 2$, $0<q<p<1$, $0<\rho\leq 1$, and $0<\varepsilon<1$, they proved that if
    \begin{equation}\label{th_eq}
        q(k-1)(e^\rho -1) +p(e^\rho-e)<\frac{k}{n}\mathrm{ln}(\varepsilon),
    \end{equation}
then the communities of $G$ induce a $\rho$-happy colouring on $G$ with probability at least $(1-\varepsilon)^n$. They defined a threshold
    \begin{equation}\label{xi}
        \xi=\max\left\{\min\left\{\ln\left(\frac{\frac{k}{n}\ln(\varepsilon)+p e +(k-1)q}{p+(k-1)q}\right),\; \frac{p}{p+(k-1)q}\right\}, \; 0\right\},
    \end{equation}
such that for $\rho\leq\xi$, with high probability, the communities of $G$ induce a complete $\rho$-happy colouring.

It was also asserted in~\cite{SHEKARRIZ2025106893} that for $0\leq \rho <\lim\limits_{n\rightarrow\infty} \xi$ and $G\in\mathcal{G}(n,k,p,q)$, when $n\rightarrow \infty$, we have $\Prob(G\in H_\rho)\rightarrow 1$. Here, we point it out that $\lim\limits_{n\rightarrow\infty} \xi=\frac{p}{p+(k-1)q}$. The reason behind this fact is that we have $$\lim_{n\rightarrow\infty} \ln\left(\frac{\frac{k}{n}\ln(\varepsilon)+p e +(k-1)q}{p+(k-1)q}\right)= \ln\left( \frac{pe+(k-1)q}{p+(k-1)q}\right).$$ Let $\tilde{\xi}=\frac{p}{p+(k-1)q}$. Then $0<\tilde{\xi}<1$, and since 
\begin{equation}\label{eq:simplified}
    \frac{p e + (k-1)q}{p + (k-1)q} = 1 + \frac{p(e-1)}{p + (k-1)q}=1+(e-1)\tilde{\xi},
\end{equation}
we have 
\begin{equation*}
    \ln\left( \frac{pe+(k-1)q}{p+(k-1)q}\right)= \ln(1+(e-1)\tilde{\xi}).
\end{equation*} Now, we have 
\begin{equation}
    \ln\left( \frac{pe+(k-1)q}{p+(k-1)q}\right)\geq \frac{p}{p+(k-1)q}
\end{equation}
because the function $f(x)=\ln(1+(e-1)x)-x$ is positive for $x\in (0,1)$. Thus, we can restate~\cite[Theorem 3.3]{SHEKARRIZ2025106893} as follows:

\begin{theorem}\label{th:infinity} \textnormal{\cite{SHEKARRIZ2025106893}}
Let $0<q<p<1$,  $k\in\mathbb{N}\setminus\{1\}$ be constants, $0\leq \rho <\tilde{\xi}=\frac{p}{p+(k-1)q}$, and $G\in\mathcal{G}(n,k,p,q)$. Then, $\Prob(G\in H_\rho) \to 1$ when $n\to \infty$. 
\end{theorem}

As a consequence of this revised statement of the theorem, when $\frac{k}{n}$ is very small, {\w that is, when} the number of vertices is much larger than the number of communities, we can take $\tilde{\xi}$ instead of $\xi$ and {\w can} expect that the graph's communities induce a $\rho$-happy colouring when $\rho\leq \tilde{\xi}$.

\section{Community detection via soft happy colouring}\label{sec:theory}

For a graph $G\in\mathcal{G}(n,k,p,q)$ with the communities of $G$ are $C_1, \ldots, C_k$, suppose that $\sigma$ is a complete $k$-colouring with colour classes $A_1, \ldots, A_k$. Recall that the proportion of vertices, whose colours in $\sigma$ agree with their original community colouring, is denoted {\w by $ACD(\sigma)$}. Therefore, we have 
\begin{equation}\label{eq:acd}
    ACD(\sigma)=\frac{1}{n}\sum_{v\in V(G)} \delta_\sigma (v),
\end{equation}
where $$\delta_\sigma (v)=\left\{\begin{aligned}
    1 & \;\;\text{ if } v\in C_{\sigma(v)}\\
    0 & \;\;\text{ otherwise}.
\end{aligned}\right.$$

Suppose that $v\in C_i$ for $1\leq i\leq k$ and $\mathrm{DEG}$ is the random variable giving the vertex degrees. Thus 
\begin{equation}\label{eq:deg}
    \mathbb{E}[\mathrm{DEG}(v)]=\left(\frac{n}{k}-1\right)p+\frac{k-1}{k}nq\approx \frac{n}{k}\left(p+(k-1)q\right).
\end{equation}

Suppose also that $D_{in}^{\sigma}$ is the random variable that  {\w denotes} the number of neighbours {\w of $v$,} whose colours agree with $\sigma(v)$. In other words, $$D_{in}^{\sigma}(v)=\sum_{u\in A_{\sigma(v)}\setminus v} \delta(u,v),$$ where 
\begin{equation}\label{eq:delta}
\delta(u,v)=\left\{\begin{aligned}
        1 & \;\;\text{ if } u\sim v\\
        0 & \;\;\text{ otherwise.} 
    \end{aligned}\right.
\end{equation}
    Therefore,
    \begin{equation}\label{eq:din}
        \mathbb{E}[D_{in}^{\sigma}(v)]=\sum_{u\in A_{\sigma(v)}\setminus v} \mathrm{Pr}(u\sim v) \approx \frac{n}{k}\left(p   \pi_v^\sigma + q  (1-\pi_v^\sigma)\right),
    \end{equation}
    where $\pi_v^\sigma$ is the expected proportion of agreement of the colour class $\sigma(v)$ and the community containing $v$. In other words, 
   \begin{equation}\label{eq:pi}\pi_v^\sigma=\frac{ \left|A_{\sigma(v)}\cap C_i\right|}{ \left| C_i\right|}=\frac{k}{n}  \left|A_{\sigma(v)}\cap C_i\right|.
   \end{equation}
Therefore,
\begin{equation}\label{eq:acd2}
    ACD(\sigma)=\frac{1}{n}\sum_{v \in V(G)} \pi_v^\sigma 
\end{equation}

Under the condition of almost equal community sizes, it is likely to compare the accuracy of community detection of two $\rho$-happy colourings that preserve a precolouring.

\begin{theorem}\label{th:better}
    Let $G$ be a partially coloured graph in $\mathcal{G}(n,k,p,q)$ with at least one vertex from each community $i$ precoloured by the colour $i$ for $1\leq i\leq k$, $n$ be sufficiently large, $k>2$ be an integer significantly small compared to $n$, and $0\leq \rho_1 < \rho_2 \leq 1$. Moreover, suppose that $\sigma_1$ and $\sigma_2$ are extensions of the precolouring to complete $k$-colourings with each of their colour classes of almost the same size as $\frac{n}{k}$, and $\sigma_1\in H_{\rho_{1}}$, $\sigma_2\in H_{\rho_{2}}$ but $\sigma_1\notin H_{\rho_{2}}$.  Then, with high probability we have $$ACD(\sigma_1)\leq ACD(\sigma_2).$$ 
\end{theorem}
\begin{proof}
If $\sigma$ is a $\rho$-happy colouring, for a vertex $v$ we have $$D_{in}^{\sigma}(v)\geq \rho\deg (v).$$ Therefore, we must have $$\mathbb{E}[D_{in}^{\sigma}(v)]\geq \rho   \mathbb{E}[\mathrm{DEG}(v)].$$ Now, by Equations~\ref{eq:din} and ~\ref{eq:deg} we have
\begin{equation*}
    \begin{split}
       & \frac{n}{k}\left(p   \pi_v^\sigma + q  (1-\pi_v^\sigma)\right) \geq \rho   \frac{n}{k}\left(p+(k-1)q\right) \\
        \implies & p   \pi_v^\sigma + q  (1-\pi_v^\sigma)  \geq  \rho  \left(p+(k-1)q\right)\\
        \implies & (p-q)\pi_v^\sigma \geq  \rho  \left(p+(k-1)q\right) -q\\
        \implies & \pi_v^\sigma \geq \frac{\rho  \left(p+(k-1)q\right) -q}{p-q}.
    \end{split}
\end{equation*}
Hence, greater $\rho$ results in greater $\pi_v^\sigma$. Now, since $\rho_1 <\rho_2$, it is inferred with high probability that $\pi_v^{\sigma_1} \leq \pi_v^{\sigma_2}$, and therefore by Equation~\ref{eq:acd2}, with high probability we have $$ACD(\sigma_1)\leq ACD(\sigma_2).$$
\end{proof}

 We know from Theorem~\ref{th:infinity} that for $G\in\mathcal{G}(n,k,p,q)$ that there is high probability that we can find a complete $\rho$-happy colouring for $G$ when $\rho\leq\tilde{\xi}$. Based on Theorem~\ref{th:better}, we expect to find higher accuracy of communities when $\rho$ is larger. An interesting question is how to find the best $\rho$ where a complete $\rho$-happy colouring has the highest degree of alignment 
 with the original communities of $G$.

In practice, a lower bound for $\rho$ is desirable so that a complete $\rho$-happy colouring acceptably matches the original community structure. One obvious lower bound is
\begin{equation}
    \rho\geq \mu = \frac{q}{p+(k-1)q}
\end{equation}
because when $\rho<\mu$, a vertex $v$ can still be $\rho$-happy if assigned to a colour other than its original community. 

Another extreme case for the proportion of happiness $\rho$ is when $\tilde{\xi}<\rho \leq 1$. Although based on Theorem~\ref{th:better}, we expect to have higher accuracy of community detection in this case, but there is a theoretical barrier that prevents this from happening. The following theorem states that for $\rho>\tilde{\xi}$, almost never an algorithm can find a complete $\rho$-happy colouring for a graph in the SBM with a large number of vertices. 

\begin{theorem}\label{th:xi-tilde}
Let $G \in \mathcal{G}(n,k,p,q)$ be a graph where $k \geq 2$ is fixed and considered small relative to $n$, and suppose that $0 < q \ll p < 1$. Assume further that the parameter $\rho$ satisfies $1 \geq \rho > \tilde{\xi} = \frac{p}{p+(k-1)q}$. Let $\sigma$ be a $k$-colouring of $G$ with colour classes $A_1, \ldots, A_k$ such that $|A_i| \approx \frac{n}{k}$ for each $i=1,\ldots,k$. Then, as $n \to \infty$, the probability that $\sigma$ contains any $\rho$-happy vertex converges to $0$. Moreover, for $n$ sufficiently large, $\sigma$ is almost never a complete $\rho$-happy colouring.
\end{theorem}
\begin{proof}
    Let the SBM communities of $G$ be $C_1,\ldots,C_k$, and $v\in C_j$ be a vertex of $G$ for $1\le j\le k$. Since $0\leq \left|A_{\sigma(v)}\cap C_i\right|\leq \frac{n}{k}$, we can deduce from Equation~\ref{eq:pi} 
    that $0<\pi_v^\sigma \leq 1$. Moreover, {\w because} 
     $q\ll p$, we have
    \begin{equation}
    \mathbb{E}[D_{in}^\sigma (v)]\leq \frac{np}{k}.
    \end{equation}
    {\w Equality} holds if $\sigma$ perfectly aligns with the community structure of $G$, i.e. if $$\{A_1,\ldots,A_k\}=\{C_1,\ldots,C_k\}.$$ 
    Figure~\ref{fig:th-comm-xi} demonstrates this observation for a graph on 120 vertices and 6 communities: for every vertex in the green community, the number of its neighbours inside the green community is less than the number of its neighbours in the original community.

    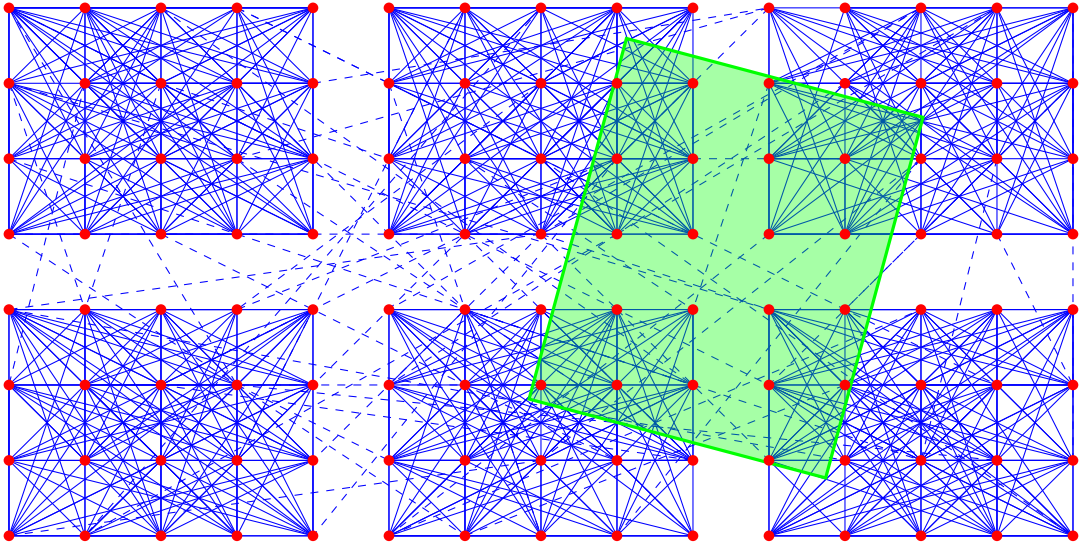
\begin{figure}
    \centering

\begin{tikzpicture}
  \def\ncommunities{6}      
  \def\nvertices{20}        
  \def\nrows{4}             
  \def\ncols{5}             
  \def\xsep{5}              
  \def\ysep{4}              
  \def\pIntra{0.7}          
  \def\pInter{0.01}         

  \foreach \i in {1,...,\ncommunities} {
      \pgfmathtruncatemacro{\crow}{int((\i-1)/3)}
      \pgfmathtruncatemacro{\ccol}{int(mod((\i-1),3))}
      \coordinate (com\i) at (\ccol*\xsep, -\crow*\ysep);
      
      \pgfmathsetmacro{\xoffset}{(\ncols-1)/2.0}
      \pgfmathsetmacro{\yoffset}{(\nrows-1)/2.0}
      
      \foreach \r in {0,...,\numexpr\nrows-1\relax} {
          \foreach \c in {0,...,\numexpr\ncols-1\relax} {
              \pgfmathtruncatemacro{\vid}{\r*\ncols + \c + 1}
              \coordinate (v\i-\vid) at ($(com\i) + (\c-\xoffset, -\r+\yoffset)$);
          }
      }
  }
  
  \begin{pgfonlayer}{background}
    \foreach \i in {1,...,\ncommunities} {
        \foreach \j in {1,...,\nvertices} {
            \ifnum\j < \nvertices
              \pgfmathtruncatemacro{\start}{\j+1}
              \foreach \k in {\start,...,\nvertices} {
                  \pgfmathsetmacro{\rando}{rnd}
                  \ifdim \rando pt < \pIntra pt
                      \draw[blue] (v\i-\j) -- (v\i-\k);
                  \fi
              }
            \fi
            \foreach \m in {1,...,\ncommunities} {
                \ifnum\m>\i
                    \foreach \l in {1,...,\nvertices} {
                        \pgfmathsetmacro{\rando}{rnd}
                        \ifdim \rando pt < \pInter pt
                            \draw[color=blue, dashed] (v\i-\j) -- (v\m-\l);
                        \fi
                    }
                \fi
            }
        }
    }
  \end{pgfonlayer}
  
  \begin{pgfonlayer}{foreground}
    \foreach \i in {1,...,\ncommunities} {
        \foreach \j in {1,...,\nvertices} {
            \fill[red] (v\i-\j) circle (2pt);
        }
    }
  \end{pgfonlayer}
  
  \filldraw[fill=green, fill opacity=0.35, draw=green, very thick, 
            rotate around={-15:(7.5,-1.5)}] (5.5,0.65) rectangle (9.54,-4.3);

\end{tikzpicture}

\caption{A graph on 120 vertices with 6 communities, each community consists of 20 vertices arranged in rectangular forms close to each other. A colour class, green rectangle, which is not aligned with the original community structure, is also shown. Evidently, the number of neighbours of a vertex in this colour class is less than the number of its neighbours in its original community. }\label{fig:th-comm-xi}
\end{figure}

Therefore, we observed that  $$\Prob(D_{in}^c (v)>\rho \mathrm{DEG}(v))\ge  \Prob(D_{in}^{\sigma} (v)>\rho \mathrm{DEG}(v)),$$ where $c$ is the colouring induced by communities.

To proceed, we need some notations from~\cite[proof of Theorem 3.1]{SHEKARRIZ2025106893}. For $i\in\{1,\ldots, k\}$, let $X_i$ be the random variable which gives the number of adjacent vertices to $v$ in $C_i$. In other words,
    \begin{equation}
        X_i=\sum_{u\in C_i} \delta(u,v),
    \end{equation}
    where $\delta(u,v)$ is introduced in Equation~\ref{eq:delta}. Therefore, 
    $$\mathbb{E}[X_i]=\begin{cases}
			p\frac{n}{k}& \text{ if } i=j \\
			q\frac{n}{k}& \text{ otherwise.}
		\end{cases}$$
    For an arbitrary $t>0$, the moment generating function of $X_i$ is $$\mathbb{E}\left[ e^{t  X_i}\right] =\begin{cases}
			\left( 1-p+p  e^t\right)^\frac{n}{k}&\text{ if } i=j \\
			\left( 1-q+q  e^t\right)^\frac{n}{k}&\text{ otherwise.}
		\end{cases}$$
    Now, we have $\mathrm{DEG}(v)=\sum_{i=1}^k X_i$ and $D_{in} (v)=D_{in}^c (v) =X_{j}$.
    
    To prove the assertion, we need to show that $\Prob(D_{in}(v)>\rho \mathrm{DEG}(v))$ approaches zero when $n\rightarrow\infty$. Suppose that $\rho=\tilde{\xi}+\varepsilon$ for $0<\varepsilon\le 1-\tilde{\xi}$. Then, by the Chernoff bound~\cite{lin2011probability}, for $t>0$ we {\w infer} 
    \begin{equation}\label{ineq:chernoff}
    \Prob(D_{in}(v)-\rho\mathrm{DEG}(v)>0)\leq \mathbb{E}\left[e^{t(D_{in}(v)-\rho \mathrm{DEG}(v))}\right].
    \end{equation}
    {\w Following an approach that is} similar to {\w that in}~\cite[proof of Theorem 3.1]{SHEKARRIZ2025106893}, {\w this result enables us to} find a bound on $\Prob(D_{in}(v)-\rho\mathrm{DEG}(v)>0)$. Using the moment generating function, we have
    $$\Prob(D_{in}(v)-\rho\mathrm{DEG}(v)>0)\leq  \left( 1-p+p  e^t\right)^{\frac{n}{k}}  \left( 1-p+p  e^{t\rho}\right)^{-\frac{n}{k}}   \prod_{i=1}^{k-1}\left( 1-q+q  e^{t\rho}\right)^{-\frac{n}{k}},$$
		Therefore,  \[\begin{split}
			\Prob(D_{in}(v)-\rho\mathrm{DEG}(v)>0)&\leq     \left( 1-p+p e^t\right)^{\frac{n}{k}}  \left( 1-p+p  e^{t\rho}\right)^{-\frac{n}{k}} \left( 1-q+q  e^{t\rho}\right)^{-\frac{k-1}{k}n}\\
			&\leq e^{-\frac{n}{k} \left( \mathrm{ln}\left(1-p+p  e^{t\rho}\right)+(k-1) \mathrm{ln}\left(1-q+q  e^{t\rho}\right)-\mathrm{ln}\left(1-p+p  e^{t}\right)\right)}.
		\end{split}\]
    {\w By} $e^x\geq 1+x$ for all $x\in\mathbb{R}$, we have $1-p+p e^{t\rho}\leq e^{p \left(e^{t\rho}-1\right)}$. Consequently, 
    
    \begin{equation}\label{ineq:pr-bound}
        \Prob(D_{in}(v)-\rho\mathrm{DEG}(v)>0)\leq e^{-\frac{n}{k}\left( p \left(e^{t\rho}-e^t\right)+(k-1)q \left(e^{t\rho}-1\right)\right)}.
    \end{equation}
    Because the right-hand side of the above inequality approaches 0 if $n\rightarrow \infty$ and
    \begin{equation}\label{ineq:th-xi}
        p  \left(e^{t\rho}-e^t\right)+(k-1)q \left(e^{t\rho}-1\right)>0,
    \end{equation}
    we only {\w have} to prove (\ref{ineq:th-xi}) for $\rho=\tilde{\xi}+\varepsilon$. This inequality is satisfied if and only if
    \begin{equation*}
        e^{tp}>\frac{pe^t +(k-1)q}{p+(k-1)q}.
    \end{equation*}
    In other words, Inequality~\ref{ineq:th-xi} is satisfied if and only if
    \begin{equation*}
        \rho> \frac{1}{t}\ln\left(\frac{pe^t +(k-1)q}{p+(k-1)q}\right),
    \end{equation*}
    which means that Inequality~\ref{ineq:th-xi} is satisfied if and only if
    \begin{equation*}
        \varepsilon> \frac{1}{t}\ln\left(\frac{pe^t +(k-1)q}{p+(k-1)q}\right)-\tilde{\xi}.
    \end{equation*}
    {\w Because } $\tilde{\xi}=\frac{p}{p+(k-1)q}$, similar to Equation~\ref{eq:simplified}, we have $\frac{pe^t +(k-1)q}{p+(k-1)q}=1+(e^t -1)\tilde{\xi}$. Therefore, 
    \begin{equation}\label{ineq:th-eps}
        \varepsilon> \frac{1}{t}\ln\left(1+(e^t -1)\tilde{\xi}\right)-\tilde{\xi},
    \end{equation} 
    {\w entails (\ref{ineq:th-xi}).} It remains to prove that for $\tilde{\xi}\in [0,1]$ and $0<\varepsilon\leq 1-\tilde{\xi}$, there is a $t>0$ such that Inequality~\ref{ineq:th-eps} holds. 
    
    Suppose that $$f(x,t)=\frac{1}{t}\ln\left(1+(e^t -1)x\right)-x.$$ We show for each $x\in [0,1]$ that $\lim\limits_{t\to 0^+} f(x,t)=0$, proving Inequality~\ref{ineq:th-eps}. We expand $f$ in a Taylor series. From the expansions
	\[
	1+(e^t-1)x = 1 + xt + \frac{x\,t^2}{2} + R_1 (t^3, x),\;\text{ and }\;
	\ln(1+y)= y - \frac{y^2}{2}+R_2 (y^3),\] 
    and $y=xt+\frac{x\,t^2}{2}+R_1 (t^3, x)$, we {\w conclude}
    \[
	f(x,t)= \frac{1}{t}\ln\Bigl(1+(e^t-1)x\Bigr)-x = \frac{(x-x^2)t}{2} + R(t^2,x),
	\]
	where $R_1$, $R_2$, and $R$ are Taylor {\w residuals} that tend to 0 when $t\to 0^+$. Thus, for every  $x\in[0,1]$, 
	\[
	\lim_{t\to 0^+} f(x,t)=0, \]
	which proves (\ref{ineq:th-eps}), and concludes that the probability of $v$ being $\rho$-happy by $c$ (and therefore by $\sigma$) approaches 0 when $n\to\infty$.

    The independence of the events of {\w an} edge between {\w a} 
    pair of vertices {\w and between any other pair of vertices, leads to}  the  inequality 
    $$\Prob(c\in H_\rho)=\left(\Prob(D_{in}(v)-\rho\mathrm{DEG}(v)>0)\right)^n \le  e^{-\frac{n^2}{k}\left( p \left(e^{t\rho}-e^t\right)+(k-1)q \left(e^{t\rho}-1\right)\right)}.$$ 
    Therefore, $\Prob(c\in H_\rho)$ should be very small (and $\Prob(\sigma\in H_\rho)$ should be even smaller) if $n$ is sufficiently large because of Inequality~\ref{ineq:th-xi}. {\w This completes the proof.} 
\end{proof}

Theorems~\ref{th:better} and \ref{th:xi-tilde} have been verified in our experimental tests; see Section~\ref{sec:verify}.

\section{Known algorithms for soft happy colouring}\label{sec:known}
Four heuristic algorithms have already been proposed for soft happy colouring of a graph with some precoloured vertices. In 2015, Zhang and Li~\cite{ZHANG2015117} proposed {\sf Greedy-SoftMHV} and {\sf Growth-SoftMHV}. In this paper, because there is no confusion, we denote these two algorithms simply by {\sf Greedy} and {\sf Growth}. The other two heuristic algorithms are {\sf LMC} and {\sf NGC}, which were introduced by Shekarriz et al.~\cite{SHEKARRIZ2025106893}.

In {\sf Greedy}, a partially coloured graph $G$ is taken as input, and tries colouring all the uncoloured vertices with one colour to see which colour generates the most $\rho$-happy vertices, which is reported as the algorithm's output. The time complexity of {\sf Greedy} is $\mathcal{O}(km)$ where $m$ is the number of edges and $k$ is the number of permissible colours, see Carpentier et al.~\cite{Carpentier2023}.

In {\sf Growth}, vertices are partitioned into {\w several} groups. {\w The} most notable ones are:
\begin{itemize}
    \item[$P$:] Already coloured vertices that can become $\rho$-happy, because there are enough uncoloured vertices adjacent to them.
    \item[$L_u$:] Vertices that are not coloured, but  will be $\rho$-unhappy, no matter {\w which} colour they receive.
    \item[$L_h$:] Vertices that are not coloured yet, but can become $\rho$-happy.
\end{itemize}
The main loop of {\sf Growth} continues as long as the set $L=P\cup L_u \cup L_h$ is not empty. Within this loop, if there {\w are $P$-verices}, it selects one, appends enough vertices from its neighbours to its colour class to make it $\rho$-happy, and recalculates the sets $P$, $L_u$, and $L_h$. If there is no $P$-vertex, an $L_h$-vertex will be chosen and {\w appended} to a colour class that makes it $\rho$-happy, {\w and then} the sets $P$, $L_u$, and $L_h$ {\w are updated}. If there is no $P$ or $L_h$ vertex, an $L_u$-vertex is chosen and coloured randomly, then  the sets $P$, $L_u$, and $L_h$ {\w are updated again}. The process repeats until all the sets $P$, $L_u$, and $L_h$ become empty, and the algorithm returns the complete colouring of the graph. The time complexity of {\sf Growth} is $\mathcal{O}(mn)$, see Carpentier et al.~\cite{Carpentier2023}.

Shekarriz et al.~\cite{SHEKARRIZ2025106893} introduced two more heuristic algorithms. One of them, namely {\sf Neighbour Greedy Colouring (NGC)}, is an extension of {\sf Greedy}. In its main loop, it tries available colours on uncoloured vertices to find which colour $i$ generates the {\w maximum}  number of $\rho$-happy vertices, but only colours neighbours of vertices that are already coloured by $i$. Tests show that {\sf NGC} and {\sf Greedy} almost always give  {\w similar solutions}, although {\sf NGC} usually {\w takes} much more time, as its time complexity is $\mathcal{O}(dkm)$, where $d=\mathrm{diam}(G)$.

{\sf Local Maximal Colouring (LMC)} is also introduced in~\cite{SHEKARRIZ2025106893} and shown to be very fast, as {\w its time complexity is} $\mathcal{O}(m)$. Moreover, its output is highly correlated with the graph's community structure. In the design of the {\sf LMC}, there is no dependence on the proportion of happiness $\rho$, but some randomness is employed. In its main loop, the algorithm randomly chooses a vertex $v$ from the intersection of uncoloured vertices and neighbours of already coloured vertices, and colours $v$ with {\w a colour that appears} the most in $N(v)$.

In order to have an idea {\w how the improvement} algorithms perform, we can look at their run results when given randomly generated feasible colourings as their inputs. Thus, we also consider {\sf Random}, {\w i.e.}  an algorithm that randomly assigns colours from $\{1,\ldots,k\}$ to the uncoloured vertices.

\section{Improvement algorithms}\label{sec:LS}

An improvement algorithm refines solutions to optimisation problems, typically by iteratively enhancing the quality of the solution. The goal is to arrive at a solution closer to optimal that satisfies a validity criterion. These algorithms are often applied when initial solutions are generated using heuristic, approximation, or other feasible methods, such as {\sf Random}. 

Inputs and outputs of our improvement algorithms for soft happy colouring are as follows:

    \begin{quote}
        \textbf{Input:} a connected graph $G$, a (probably partial) $k$-colouring $\sigma$, the proportion of happiness $\rho$, and the set of precoloured vertices $V'$. 
        
        \textbf{Output:} a $k$-colouring $\tilde{\sigma}$ such that $H_\rho (\tilde{\sigma})\geq H_\rho (\sigma)$.
    \end{quote}

We propose three local search improvement algorithms for soft happy colouring, namely {\sf Local Search} ({\sf LS} for short), {\sf Repeated Local Search} ({\sf RLS} for short) and {\sf Enhanced Local Search} ({\sf ELS} for short). All of these algorithms may come with the suffix ``{\sf -SoftMHV}'', but because there is no confusion here, we do not mention the suffix. The algorithms receive a (probably partial) $k$-colouring $\sigma$ and transform it into a $k$-colouring $\tilde{\sigma}$ so that the number of $\rho$-happy vertices in the latter is greater than or equal to the former.

The improvement algorithm \hyperref[alg:ls]{\sf Local Search (LS)} is presented in detail in Algorithm~\ref{alg:ls}. It improves the number of $\rho$-happiness as follows. First, in Line 1, the colouring $\tilde{\sigma}$ is initially the same as the input colouring $\sigma$. Then, in Line 2, all the $\rho$-unhappy vertices that are not precoloured are gathered in the set $U$. Afterwards, in the algorithm's main loop, Lines 3--8, all the vertices in $U$ are examined; for a vertex $u\in U$, if $u$ has at least one coloured neighbour, its colour becomes $q$, {\w a colour} which appears the most in $N(u)$. 

If there is a vertex $u$ in $U$ such that all {\w of} its neighbours are {\w uncoloured} in $\sigma$, then no colour $q$ (Line 5) appears the most in $N(u)$. So, at Line 7, it is possible that $U\neq \emptyset$. Because some vertices of $U$ are already coloured, the remaining vertices in $U$ can be checked repeatedly until $U$ becomes empty. Moreover, due to the assumption of $G$ being connected, the main loop ends.

It is also possible that while the algorithm makes a vertex $\rho$-happy, it makes some of its neighbours $\rho$-unhappy. Therefore, it is essential to check at the completion of the algorithm, Lines 9--11, whether it improved the $\rho$-happiness vertices. If it does not enhance the number of $\rho$-happy vertices, then the algorithm reports the input colouring, ensuring that we always have $$H_\rho (\tilde{\sigma})\geq H_\rho (\sigma).$$ 

\begin{algorithm}[!ht]
    \caption{(Local Search (LS)) Improving the solution of $\rho$-happy colouring.}\label{alg:ls}
   
    \begin{algorithmic}[1]

        \State $\tilde{\sigma}\gets \sigma$
       
        \State $U\gets \rho$-unhappy free vertices \Comment{$U=\{v\in V\setminus V'\;: \; v \text{ is }\rho\text{-unhappy in } c \}$}
        \NoNumber{ }
        \While{$U\neq \emptyset$}
        
        \For{$u\in U$}
        
        \State $\tilde{\sigma}(u)\gets$ the colour which appears the most in $N(u)$
        \State {\bf Remove} $u$ from $U$
       
        \EndFor
        
        \EndWhile
       
        \If{$H_\rho (\tilde{\sigma})<H_\rho (\sigma)$}
        \State $\tilde{\sigma} \gets \sigma$
        \EndIf
        \State {\bf Return} $\tilde{\sigma}$
    \end{algorithmic}
\end{algorithm}

\hyperref[alg:ls]{\sf LS} is a fast algorithm whose time complexity is linear in terms of the number of edges, $\mathcal{O}(m)$. The algorithm serves two purposes: (a) improves upon the solution provided by another algorithm, and (b) generates a heuristic solution if it is given a partial colouring as input. 

Another heuristic and improvement algorithm, namely \hyperref[alg:rls]{\sf Repeated Local Search (ELS)}, which is presented in Algorithm~\ref{alg:rls}, extends \hyperref[alg:ls]{\sf LS} by changing $U$ to the set of $\rho$-unhappy free vertices in Line~8. This increases the algorithm's complexity, so it is essential to check if the set $U$ actually changes, which is done in Lines~9 to~11, and if not, the main loop ends. The complexity of calculating $H_\rho (\tilde{\sigma})$ for all vertices of $G$ is $\mathcal{O}(m)$, computed at Line~8 of \hyperref[alg:rls]{\sf RLS}. Therefore, the time complexity of \hyperref[alg:rls]{\sf RLS} is $\mathcal{O}(m^2)$.

\begin{algorithm}
    \caption{(Repeated Local Search (RLS)) Improving the solution of $\rho$-happy colouring.}\label{alg:rls}
    \begin{algorithmic}[1]

        \State $\tilde{\sigma}\gets \sigma$
        
        \State $U\gets \rho$-unhappy free vertices \Comment{$U=\{v\in V\setminus V'\;: \; v \text{ is }\rho\text{-unhappy in } c \}$} 
        \NoNumber{ }
        \While{$U\neq \emptyset$}
       
        \For{$u\in U$}
       
        \State $\tilde{\sigma}(u)\gets$ the colour which appears the most in $N(u)$
        \State {\bf Remove} $u$ from $U$
       
        \EndFor
        \State $U\gets \rho$-unhappy free vertices 
        \If{$U$ is repeated to be the same set} 
        \State {\bf Break While}
        \EndIf
        \EndWhile
        
        \If{$H_\rho (\tilde{\sigma})<H_\rho (\sigma)$}
        \State $\tilde{\sigma} \gets \sigma$
        \EndIf
        \State {\bf Return} $\tilde{\sigma}$
    \end{algorithmic}
\end{algorithm}

The third proposed heuristic algorithm, \hyperref[alg:els]{\sf Enhanced Local Search (PLS)}, is a conventional optimisation local search algorithm that checks the objective function $H_\rho (\tilde{\sigma})$ at every step of its main loop.  It does not check the colour that appears the most in $N(u)$, instead, in Line~6, it checks the colours that are able to make $u$ a $\rho$-happy vertex and picks the one that generates the largest number of $\rho$-happy vertices in $G$. As said, the complexity of calculating $H_\rho (\tilde{\sigma})$ for all vertices of $G$ is $\mathcal{O}(m)$. Therefore, the time complexity of \hyperref[alg:els]{\sf ELS} is $\mathcal{O}(mnk)$.

\begin{algorithm}
    \caption{(Enhanced Local Search (ELS)) Improving the solution of $\rho$-happy colouring.}\label{alg:els}
    \begin{algorithmic}[1]
        
        \State $\tilde{\sigma}\gets \sigma$
       
        \State $U\gets \rho$-unhappy vertices \Comment{$U=\{v\in V\setminus V'\;: \; v \text{ is }\rho\text{-unhappy in } c \}$}
        \NoNumber{ }
       
        \For{$u\in U$}
        \State $c_t \gets \tilde{\sigma}$\Comment{$c_t$ is the best solution found so far}
        \For{$q\in \{1,\ldots,k\}$}
        \If{$\vert \{ v\in N(u) \; : \; \sigma(v)=q\} \vert \geq \rho   \deg (v)$}
        \State $c_t (u)\gets q$
        \If{$H_\rho (\tilde{\sigma})<H_\rho (c_t)$}
        \State $\tilde{\sigma}(u) \gets q$
        \State {\bf Remove} $u$ from $U$
        \EndIf
        \EndIf
        \EndFor
        \EndFor
       
        \State {\bf Return} $\tilde{\sigma}$
    \end{algorithmic}
\end{algorithm}

\section{Experimental evaluation}\label{sec:analys}
\subsection{Tests details}\label{sec:details}
To test the heuristic and improvement algorithms, we use the same set of 28,000 randomly generated partially coloured graphs in the SBM, introduced in~\cite{SHEKARRIZ2025106893}. The graphs are stored in DIMACS format and are publicly available\footnote{at \href{https://github.com/mhshekarriz/HappyColouring_SBM}{https://github.com/mhshekarriz/HappyColouring\_SBM}}. The graphs consist of $200 \leq n < 3,000$ vertices. For each $n$, 10 instances are available, a total of 28,000 randomly generated partially coloured graphs. For each graph, the other parameters are randomly chosen within the intervals as follows: $k \in \{2,3,\ldots,20\}$, $p \in (0,1]$, $q \in (0, \frac{p}{2}]$, and $\rho \in (0,1]$. The time limit for our tests is set to 60 seconds.

The number of precoloured vertices per community for the 28,000 graphs is $1\leq pcc \leq 10$, which sometimes seems relatively low compared to the total number of vertices. Therefore, the number of free vertices of these instances can reach 2998 for a graph with 3000 vertices, two communities, and one precoloured vertex per community. To study the effects of higher $pcc$ requires a new set of generated graphs, which will be discussed in Section~\ref{sec:pcc}.

\subsection{Verification of theoretical results}\label{sec:verify}
To test out theoretical results, Theorem~\ref{th:better} and~\ref{th:xi-tilde}, we checked whether they hold for the heuristic algorithms (and the improvement algorithms when they are given the original partially coloured graphs as their inputs). First, we check how changes in $\rho$ affect the accuracy of community detection of an algorithm when it produces a complete $\rho$-happy colouring. To do this, for each tested graph, we split the interval $[0,1]$ into three subintervals: $[0,\mu)$, $[\mu, \tilde{\xi}]$, and $(\tilde{\xi}, 1]$. The result of this setting is reported in Figure~\ref{fig:th-comm1}. 

\begin{figure}
    \centering
    \includegraphics[scale=0.48]{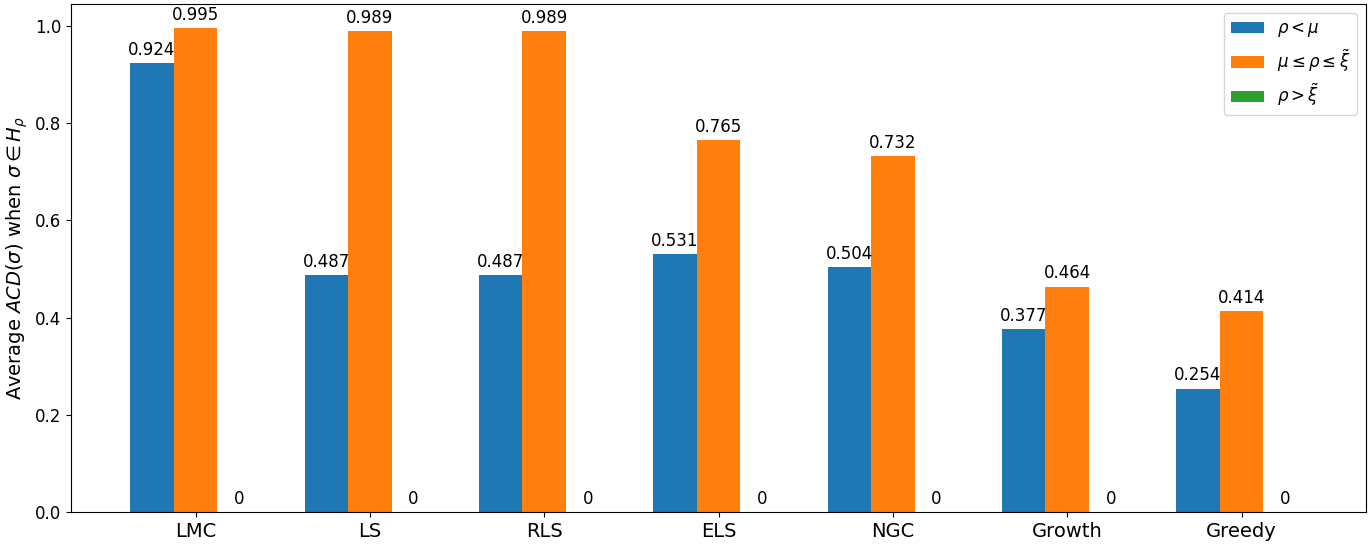}
    \caption{Average accuracy of community detection of algorithms (i.e. average $ACD(\sigma)$), when they produce a complete $\rho$-happy colouring. When $\tilde{\xi}<\rho\leq 1$, no algorithm could find a complete $\rho$-happy colouring; therefore, the respective value for each algorithm is 0.}
    \label{fig:th-comm1}
\end{figure}

As can be seen in Figure~\ref{fig:th-comm1}, the assertion of Theorem~\ref{th:better} holds for all the heuristic algorithms. The accuracies of community detection of {\sf LMC}, \hyperref[alg:ls]{\sf LS}, and \hyperref[alg:rls]{\sf RLS} are near-complete when their output is a complete $\rho$-happy colouring and $\rho\in [\mu,\tilde{\xi}]$. Further details are provided in Table~\ref{table:theory}. The table shows the number of times each algorithm could (or could not) find a complete $\rho$-happy colouring. Moreover, both Figure~\ref{fig:th-comm1} and Table~\ref{table:theory} verify Theorem~\ref{th:xi-tilde} as no algorithm can find a complete happy colouring when $\tilde{\xi}<\rho\leq 1$.

It must also be noted that although in our tests the behaviours of {\sf Greedy} and {\sf NGC} showed to be in accordance with Theorem~\ref{th:better}, it is possible that in another test they behave differently because the colour classes of their outputs seldom have almost the same size.  

\begin{table}
\centering
\begin{NiceTabular}{|c|cc|cc|cc|}[hvlines, code-before =
\rectanglecolor{vibtilg!50}{1-1}{1-7}
\rectanglecolor{des!50}{2-1}{2-7}
\rectanglecolor{C_LMC!40}{3-1}{3-7}
\rectanglecolor{C_LS!50}{4-1}{4-7}
\rectanglecolor{C_ELS!35}{5-1}{5-7}
\rectanglecolor{C_PLS!50}{6-1}{6-7}
\rectanglecolor{C_NGC!40}{7-1}{7-7}
\rectanglecolor{C_growth!40}{8-1}{8-7}
\rectanglecolor{C_greedy!50}{9-1}{9-7}
\rectanglecolor{des!50}{10-1}{11-7}
]
\hline
{\bf Condition} & \multicolumn{2}{c|}{$\boldsymbol{0\leq \rho<\mu}$} & \multicolumn{2}{c|}{$\boldsymbol{\mu\leq \rho\leq\tilde{\xi}}$} & \multicolumn{2}{c|}{$\boldsymbol{\tilde{\xi}< \rho\leq 1}$} \\
\hline
Description  & \# $c\in H_\rho$  & \# $c\notin H_\rho$  & \# $c\in H_\rho$  & \# $c\notin H_\rho$  & \# $c\in H_\rho$    & \# $c\notin H_\rho$   \\
\hline
{\sf LMC}  & 635  & 1717  & 1241  & 7166  & 0  & 17241   \\
\hline
{\sf LS} & 59  & 2293  & 63  & 8344  & 0  & 17241  \\
\hline
{\sf RLS}  & 59  & 2293  & 129  & 8278  & 0  & 17241   \\
\hline
{\sf ELS} & 38  & 2314  & 21  & 8386  & 0  & 17241  \\
\hline
{\sf NGC}  & 516  & 1836  & 78  & 8329  & 0  & 17241   \\
\hline
{\sf Growth} & 837  & 1515  & 407  & 8000  & 0  & 17241   \\
\hline
{\sf Greedy}  & 68  & 2284  & 3  & 8404  & 0  & 17241   \\
\hline
\Block{2-1}{Totals} &  \multicolumn{2}{c|}{2352} &  \multicolumn{2}{c|}{8407} &  \multicolumn{2}{c|}{17241}\\
 &  \multicolumn{6}{c|}{28000}\\ 
 \hline
\end{NiceTabular}

\medskip
\caption{The number of times each algorithm finds a complete $\rho$-happy colouring (denoted by \# $c\in H_\rho$) versus the number of times each algorithm fails to find a complete $\rho$-happy colouring (denoted by \# $c\notin H_\rho$) when $\rho \in [0,\mu)$, $\rho \in [\mu, \tilde{\xi}]$, or $\rho \in (\tilde{\xi},1]$. When $\rho\in (\tilde{\xi},1]$, no algorithm could find a complete $\rho$-happy colouring, so the respective value is 0. The row sum for each row is 28,000, the number of tested graphs.}
\label{table:theory}
\end{table}

Figure~\ref{fig:th-comm2} illustrates how changes in $\rho$ affect the accuracies of community detection of the algorithms in general, i.e., when they do not necessarily produce a complete $\rho$-happy colouring. The figure shows average values of the accuracies considering the three subintervals.

\begin{figure}
    \centering
    \includegraphics[scale=0.48]{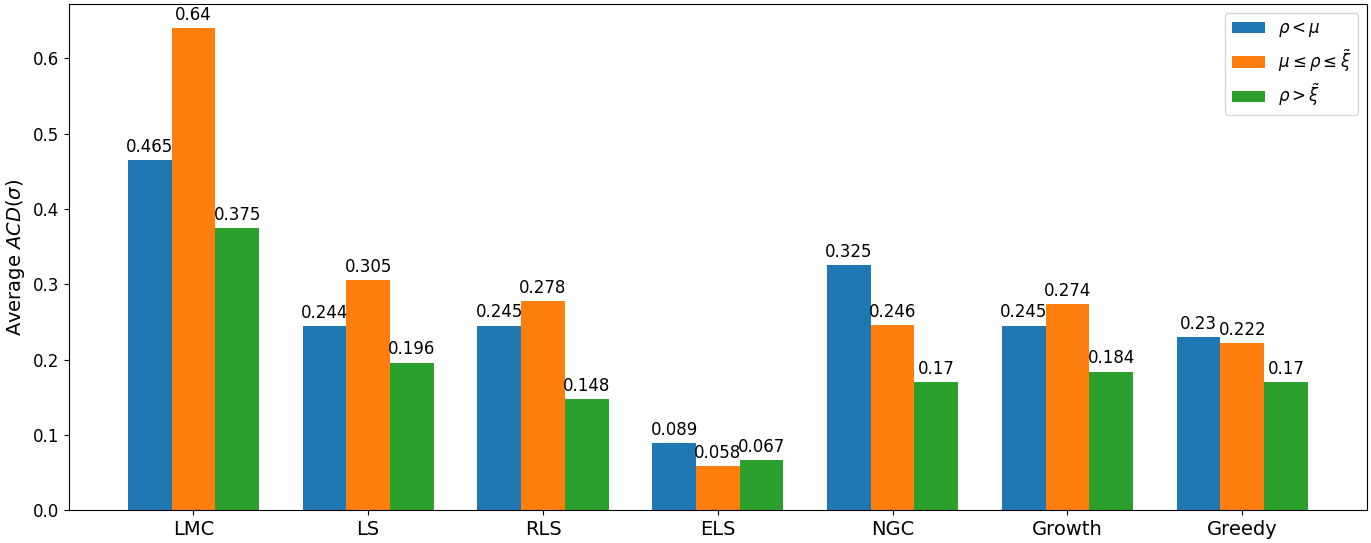}
    \caption{Average community detection accuracy (i.e. $ACD(\sigma)$) of the colouring outputs of the tested algorithms when completeness of $\rho$-happy colouring is relaxed.}
    \label{fig:th-comm2}
\end{figure}

It can be seen in Figure~\ref{fig:th-comm2} that for the algorithms {\sf LMC}, \hyperref[alg:ls]{\sf LS}, \hyperref[alg:rls]{\sf RLS}, and {\sf Growth}, higher accuracy of community detection can be seen when $\mu\leq \rho \leq \tilde{\xi}$. However, this is not the case for {\sf NGC} and {\sf Greedy} as these algorithms seldom produce equal-sized partitions. For \hyperref[alg:els]{\sf ELS}, the accuracy of community detection is not reliable due to very poor performance caused by its high time-complexity, see Section~\ref{sec:analys-heu}. 

\subsection{Analysis of heuristic algorithms}\label{sec:analys-heu}

On average, {\sf Greedy} finds the highest ratio of $\rho$-happy vertices among the known heuristics while {\sf NGC} follows closely behind (Figure~\ref{fig:bar-happy-heuristics}). Both these algorithms could find colourings with more than \%90 $\rho$-happy vertices. This average value is \%87 for \hyperref[alg:ls]{\sf LS}, \%71 for {\sf LMC}, \%68 for {\sf Growth}, and \%53 for \hyperref[alg:rls]{\sf RLS}. In general, \hyperref[alg:els]{\sf ELS} could not generate full colouring in the time limit of 60 seconds, so the quality of solutions generated by it is only \%2.

The above result is not surprising since the number of precoloured vertices in our tests was relatively low compared to the total number of vertices. However, if the number of precoloured vertices increases, we can expect the ratio of $\rho$-happy vertices for {\sf Greedy} and {\sf NGC} to drop significantly, while other algorithms perform even better. A study dedicated to such a scenario is presented in Section~\ref{sec:pcc}. Even with the problem instances consist of low numbers of precoloured vertices, \hyperref[alg:ls]{\sf LS} is able to outperform all other heuristics when $\rho<\xi$ (Figure~\ref{fig:bar-happy-heuristics-xi}). In this case, {\sf Greedy}, {\sf NGC} and {\sf LMC} follow \hyperref[alg:ls]{\sf LS}, all generating high quality solutions.

\begin{figure}
    \centering
    \includegraphics[scale=0.6]{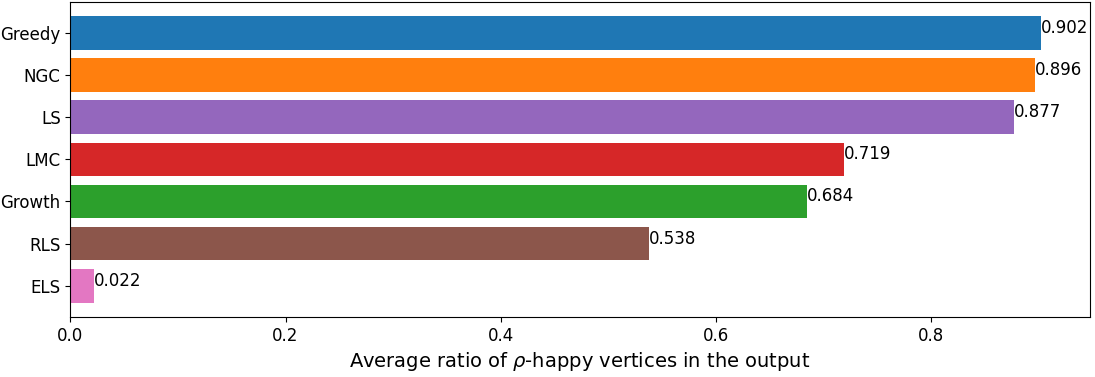}
    \caption{Comparing the average ratio of $\rho$-happy vertices in the output of the algorithms. Here, all the algorithms are given the raw pre-coloured graphs as inputs.}
    \label{fig:bar-happy-heuristics}
\end{figure}

\begin{figure}
    \centering
    \includegraphics[scale=0.6]{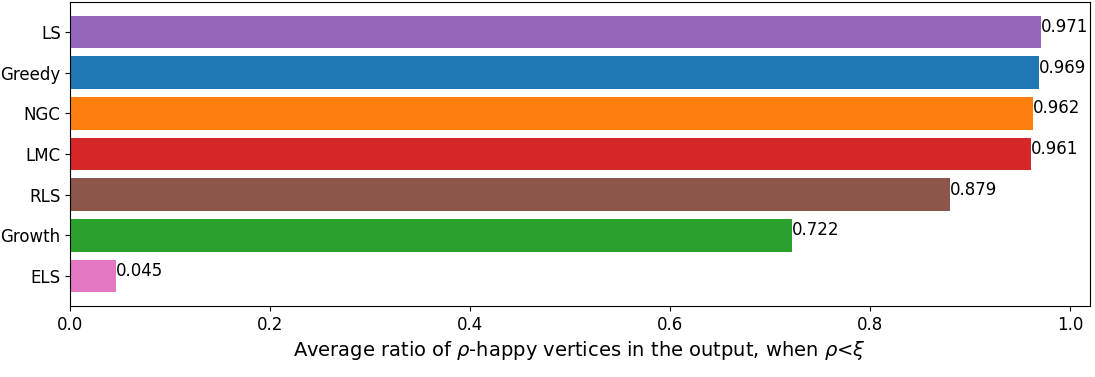}
    \caption{Comparing the average ratio of $\rho$-happy vertices in the output of the algorithms, when $\rho<\xi$. Here, all the algorithms are given the raw pre-coloured graphs as inputs.}
    \label{fig:bar-happy-heuristics-xi}
\end{figure}

Concerning the accuracy of community detection, {\sf LMC} performs best by far (Figure~\ref{fig:bar-comm-heuristics}). \hyperref[alg:ls]{\sf LS} follows next, but is only able to attain half the accuracy of {\sf LMC}.

\begin{figure}
    \centering
    \includegraphics[scale=0.6]{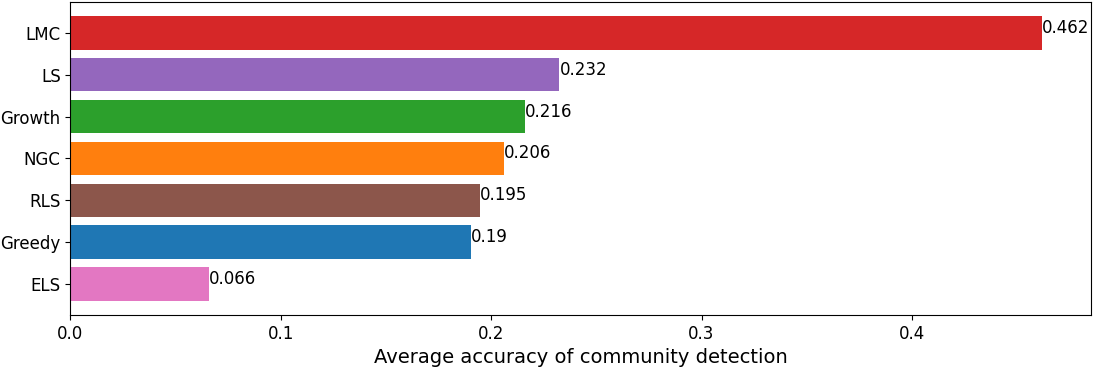}
    \caption{Comparing the average accuracy of community detection in the output of the algorithms. Here, all the algorithms are given the raw pre-coloured graphs as inputs.}
    \label{fig:bar-comm-heuristics}
\end{figure}

Concerning run-times, {\sf LMC}, \hyperref[alg:ls]{LS} and \hyperref[alg:rls]{\sf RLS} have the lowest average running time among the heuristics (Figure~\ref{fig:bar-time-heuristics}). {\sf ELS} requires the longest average run-times, almost twice as those of {\sf Growth}. 

\begin{figure}
    \centering
    \includegraphics[scale=0.6]{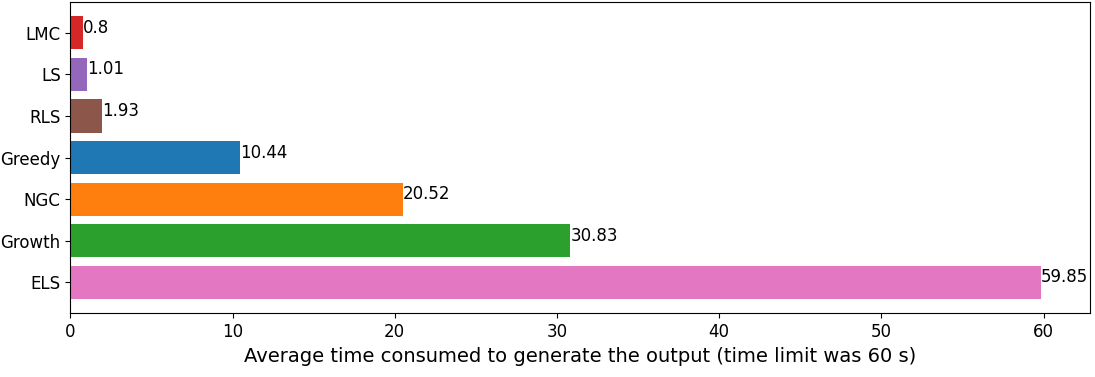}
    \caption{Comparing the average running time of the algorithms. Here, all the algorithms are given the raw pre-coloured graphs as inputs.}
    \label{fig:bar-time-heuristics}
\end{figure}

Figure~\ref{fig:heu-hap} compares the ratio of $\rho$-happy vertices of the heuristics, over the change of the number of vertices $n$, the proportion of happiness $\rho$ and the number of colours $k$. Figure~\ref{fig:heu-hap-n} shows that when the number of vertices grows, the average ratio of $\rho$-happy vertices in the output of the \hyperref[alg:ls]{LS} increases. Although {\sf Greedy} and {\sf NGC} still surpass \hyperref[alg:ls]{LS}, the difference between them becomes negligible when $n$ is large enough. Due to its increased resource requirements, it is not unexpected that the average ratio of $\rho$-happy vertices produced by {\sf Growth} drops when $n$ grows.

\begin{figure}
\captionsetup{size=small}
\begin{subfigure}{0.5\textwidth}
    \includegraphics[scale=0.5]{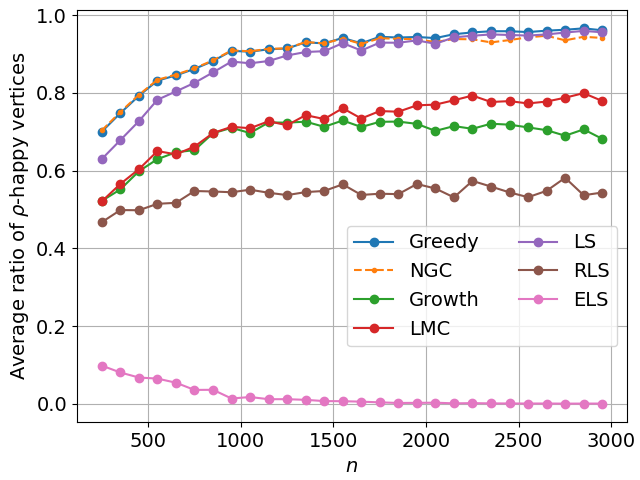}
    \caption{}\label{fig:heu-hap-n}
\end{subfigure} 
\hspace{5mm}
\begin{subfigure}{0.5\textwidth}
    \includegraphics[scale=0.5]{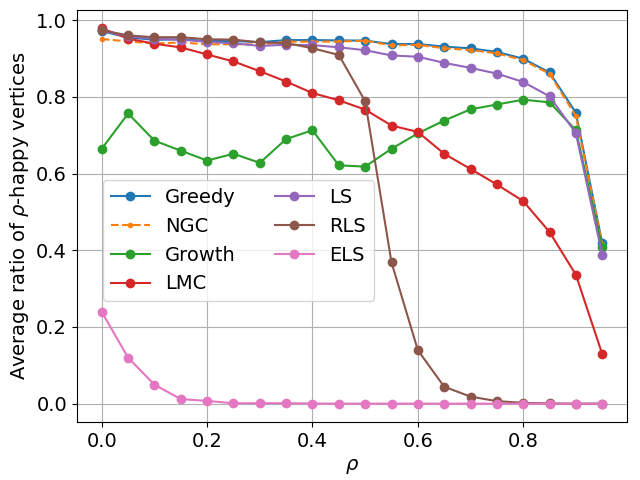}
    \caption{}\label{fig:heu-hap-r}
\end{subfigure} 

\begin{subfigure}{1\textwidth}
\centering
    \includegraphics[scale=0.5]{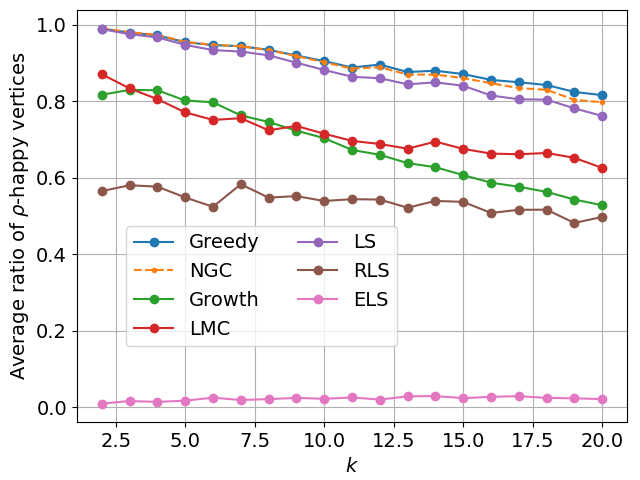}
    \caption{}\label{fig:heu-hap-k}
\end{subfigure} 

\caption{Comparison of the known heuristic algorithms for their average number of $\rho$-happy vertices considering (a) the number of vertices $n$ (b) the proportion of happiness $\rho$ and (c) the number of colours $k$.}\label{fig:heu-hap}
\end{figure}

Figure~\ref{fig:heu-hap-r} demonstrates the superiority of \hyperref[alg:ls]{LS} over the other heuristics when $\rho\leq 0.4$. It is only for $\rho>0.4$ that outputs of {\sf Greedy} and {\sf NGC} demonstrate higher ratios of $\rho$-happy vertices. The heuristics, including \hyperref[alg:ls]{LS}, experience sharp falls when $\rho$ approaches 1. Furthermore, on average, the ratios of $\rho$-happy vertices drop when the number of colours grows across all algorithms (Figure~\ref{fig:heu-hap-k}). However, \hyperref[alg:ls]{LS} is able to maintain a ratio higher than 0.75 when $k\leq 20$, slightly below {\sf Greedy} and {\sf NGC}. 

The CPU run-times of the tested heuristic algorithms are shown in Figure~\ref{fig:heu-time}. The two linear-time heuristics, \hyperref[alg:ls]{LS} and {\sf LMC}, perform almost similarly, hardly reaching 5 s for graphs on 3,000 vertices (Figure~\ref{fig:heu-time-n}). Supported by Figure~\ref{fig:heu-time-r}, changes of $\rho$ have no dramatic effect on the time consumed by the heuristics. Increasing the number of colours results in increases in running times of {\sf Growth} and {\sf NGC}, while slightly reducing those of \hyperref[alg:ls]{\sf LS} and {\sf LMC} (Figure~\ref{fig:heu-time-k}). In general, the running times of \hyperref[alg:ls]{\sf LS} and {\sf LMC} are much lower than the other heuristics. 

\begin{figure}
\captionsetup{size=small}
\begin{subfigure}{0.5\textwidth}
    \includegraphics[scale=0.5]{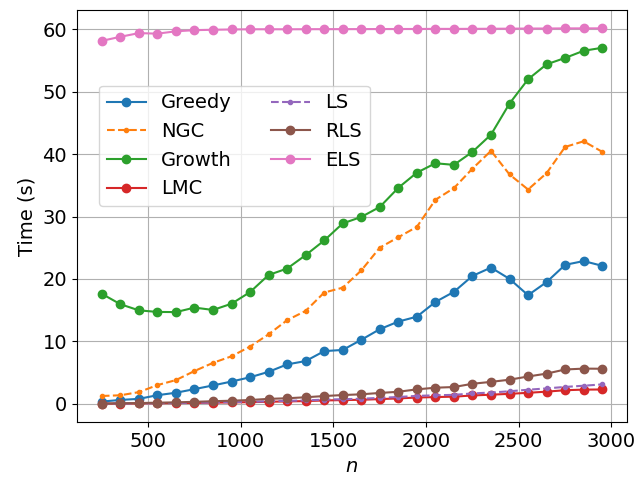}
    \caption{}\label{fig:heu-time-n}
\end{subfigure} 
\hspace{5mm}
\begin{subfigure}{0.5\textwidth}
    \includegraphics[scale=0.5]{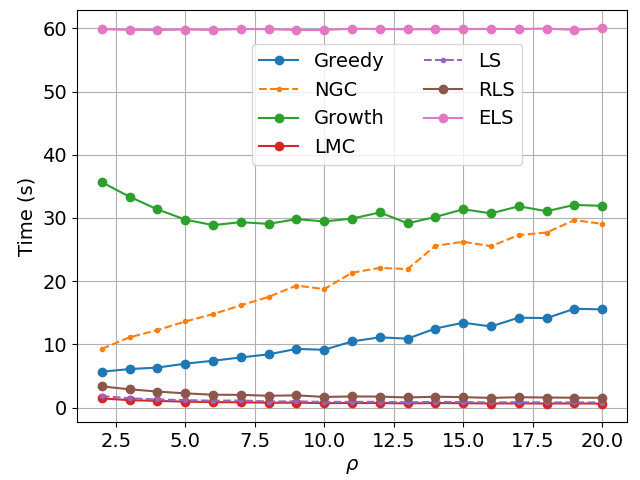}
    \caption{}\label{fig:heu-time-r}
\end{subfigure} 

\begin{subfigure}{1\textwidth}
\centering
    \includegraphics[scale=0.5]{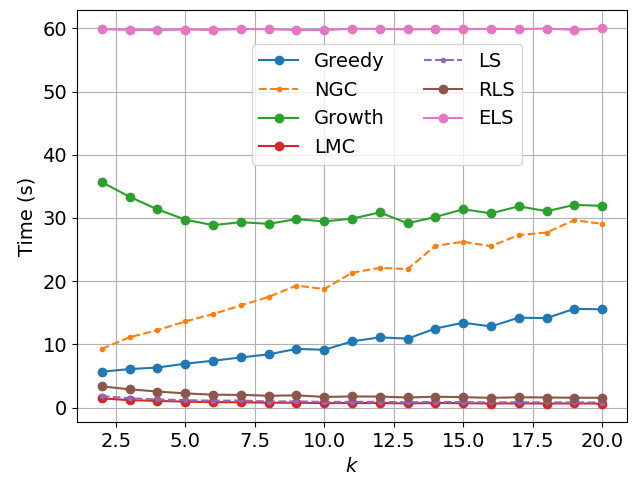}
    \caption{}\label{fig:heu-time-k}
\end{subfigure} 

\caption{Comparison of the known heuristic algorithms for their running times concerning (a) the number of vertices $n$, (b) the proportion of happiness $\rho$ and (c) the number of colours $k$.}\label{fig:heu-time}
\end{figure}

For the accuracy of community detection, Figure~\ref{fig:heu-comm} shows that \hyperref[alg:ls]{\sf LS} outperforms the other algorithms, except {\sf LMC}. The accuracy of community detection decreases when $n$ or $k$ grow (Figures~\ref{fig:heu-comm-n} and \ref{fig:heu-comm-k}), with no significant visible changes occurring when $\rho$ changes (Figure~\ref{fig:heu-comm-r}).

\begin{figure}
\captionsetup{size=small}
\begin{subfigure}{0.5\textwidth}
    \includegraphics[scale=0.5]{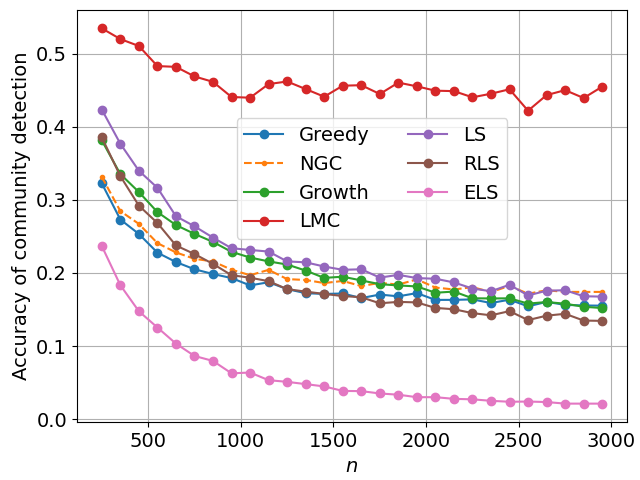}
    \caption{}\label{fig:heu-comm-n}
\end{subfigure} 
\hspace{5mm}
\begin{subfigure}{0.5\textwidth}
    \includegraphics[scale=0.5]{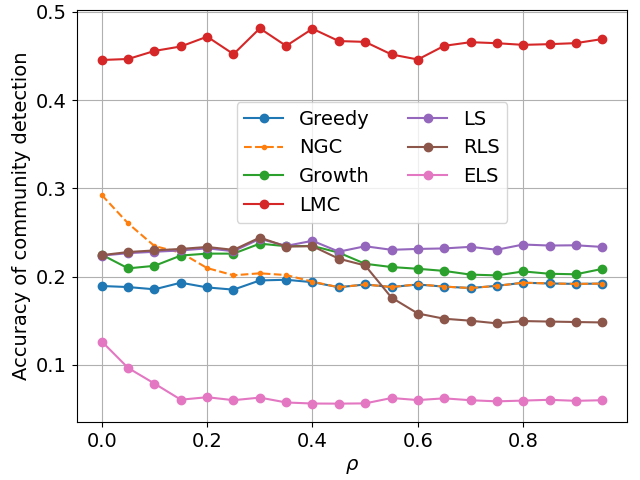}
    \caption{}\label{fig:heu-comm-r}
\end{subfigure} 

\begin{subfigure}{1\textwidth}
\centering
    \includegraphics[scale=0.5]{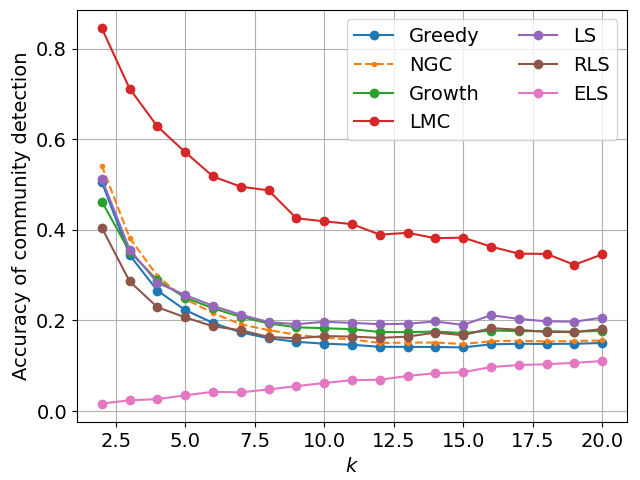}
    \caption{}\label{fig:heu-comm-k}
\end{subfigure} 

\caption{Comparison of the known heuristic algorithms for their accuracy of community detection considering (a) the number of vertices $n$ (b) the proportion of happiness $\rho$ and (c) the number of colours $k$.}\label{fig:heu-comm}

\end{figure}

\subsection{Analysis of improvement heuristics}

As mentioned, \hyperref[alg:ls]{LS}, \hyperref[alg:rls]{\sf RLS} and \hyperref[alg:els]{\sf ELS} are improvement algorithms, and can hence improve upon solutions produced by a heuristic algorithm. While outputs of {\sf LMC} and {\sf growth} can be enhanced by an improvement algorithm, the outputs of {\sf Greedy} and {\sf NGC} are unlikely to feed such improvements in $\rho$-happy colouring because almost all of their free vertices receive only one colour. Therefore, if given those solutions as inputs, the local search improvement algorithms improve nearly no $\rho$-happiness especially when the proportion of free vertices is high. Consequently, we exclude these two algorithms when considering improvements to their solutions.

The time limit of 60 seconds is often exhausted for \hyperref[alg:els]{\sf ELS} and {\sf Growth}), in which case, an incomplete colouring is reported. Then the three improvement algorithms run over the generated solutions. The time limit for this part is another 60 s (a total of 120 s). It must be noted that we have not considered running a lighter improvement algorithm over a solution of the same or heavier improvement algorithms. Therefore, we report no result of running \hyperref[alg:ls]{\sf LS} over solutions of \hyperref[alg:rls]{\sf RLS}. 

Figure~\ref{fig:imp-lg} illustrates the trends of the average ratio of $\rho$-happy vertices in solutions generated by {\sf LMC} (red) and {\sf Growth} (green) and their improved solutions by \hyperref[alg:ls]{\sf LS} (purple), \hyperref[alg:rls]{\sf RLS} (brown), and \hyperref[alg:els]{\sf ELS} (pink). Figures~\ref{fig:LMC_n_improved}, \ref{fig:LMC_r_improved} and \ref{fig:LMC_k_improved} respectively demonstrate the trends of this improvement for solutions of {\sf LMC} considering the number of vertices $n$, proportion of happiness $\rho$ and number of colours $k$. In contrast, Figures~\ref{fig:Growth_k_improved}, \ref{fig:Growth_r_improved} and \ref{fig:Growth_k_improved} show the similar trends for {\sf Growth}. Figures~\ref{fig:LMC_n_improved}, \ref{fig:LMC_r_improved}, and \ref{fig:LMC_k_improved} show that the improvement algorithms, \hyperref[alg:ls]{\sf LS} (purple) and \hyperref[alg:rls]{\sf RLS} (brown), on average, show similar performance when combined with {\sf LMC}. However, their performance is significantly different on solutions generated by {\sf Growth}, as Figures~\ref{fig:Growth_n_improved}, \ref{fig:Growth_r_improved}, and \ref{fig:Growth_k_improved}, clearly show that \hyperref[alg:rls]{\sf RLS} produces considerable improvements. \hyperref[alg:ls]{\sf LS} and \hyperref[alg:rls]{\sf RLS} typically provide better improvements compared to \hyperref[alg:els]{\sf ELS} (pink) (Figure~\ref{fig:imp-lg}). As the number of vertices increases, \hyperref[alg:els]{\sf ELS} seems to only be able to make small improvements (Figures~\ref{fig:LMC_n_improved} and \ref{fig:Growth_n_improved}). One possible explanation is due to the large time requirements of \hyperref[alg:els]{\sf ELS}, which often exhausts the time limit when there are 500 or more vertices (Figure~\ref{fig:heu-time-n}).

\begin{figure}
\captionsetup{size=small}
\begin{subfigure}{0.48\textwidth}
    \includegraphics[scale=0.48]{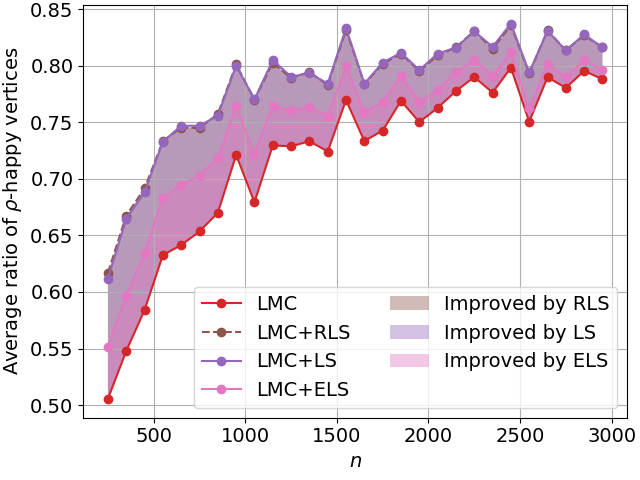}
    \caption{}\label{fig:LMC_n_improved}
\end{subfigure} 
\hspace{5mm}
\begin{subfigure}{0.48\textwidth}
    \includegraphics[scale=0.48]{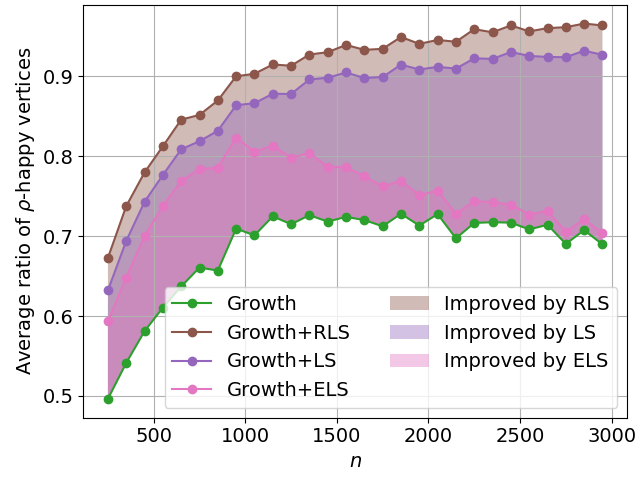}
    \caption{}\label{fig:Growth_n_improved}
\end{subfigure} 

\begin{subfigure}{0.48\textwidth}  
    \includegraphics[scale=0.48]{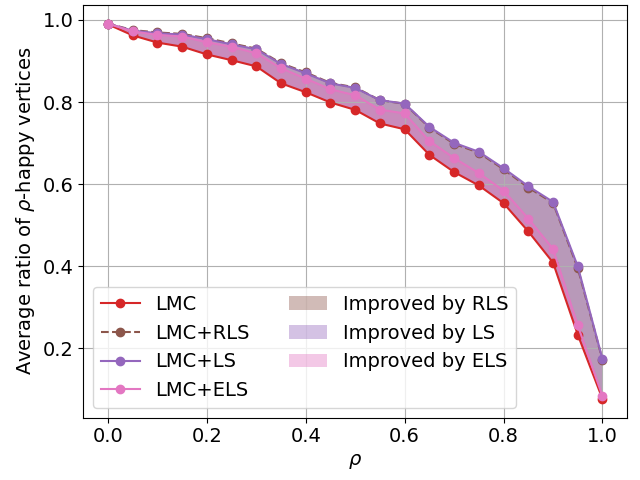}
    \caption{}\label{fig:LMC_r_improved}
\end{subfigure} 
\hspace{5mm}
\begin{subfigure}{0.48\textwidth}  
    \includegraphics[scale=0.48]{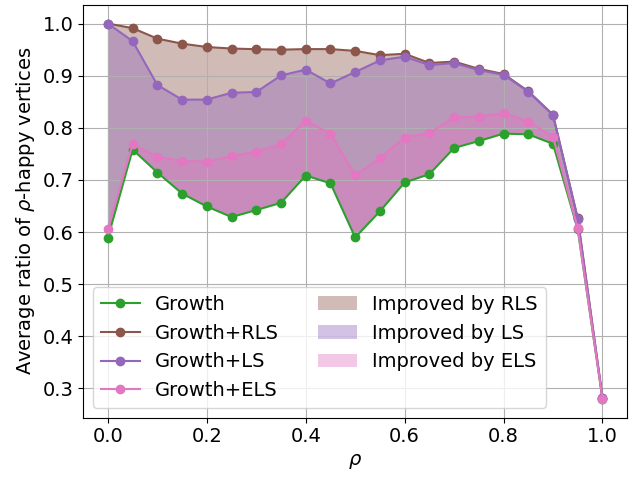}
    \caption{}\label{fig:Growth_r_improved}
\end{subfigure}

\begin{subfigure}{0.48\textwidth}  
    \includegraphics[scale=0.48]{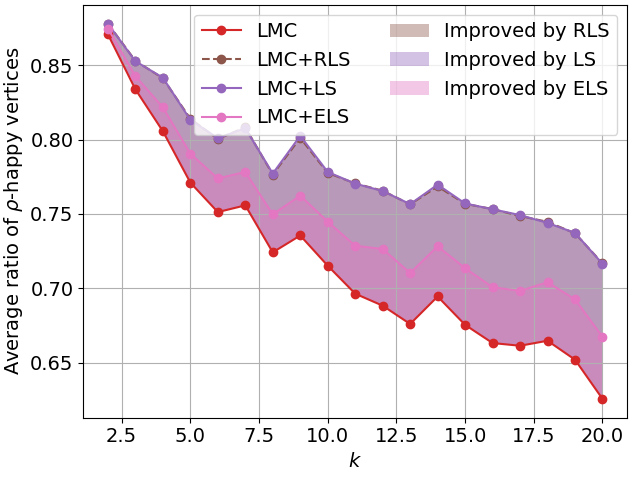}
    \caption{}\label{fig:LMC_k_improved}
\end{subfigure} 
\hspace{5mm}
\begin{subfigure}{0.48\textwidth}  
    \includegraphics[scale=0.48]{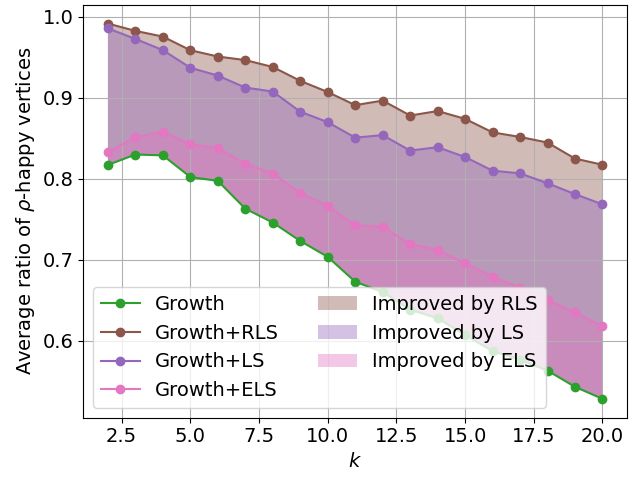}
    \caption{}\label{fig:Growth_k_improved}
\end{subfigure} 
    
    \caption{Results of improvements by \hyperref[alg:ls]{\sf LS} (Algorithm~\ref{alg:ls}) on the ratio of $\rho$-happy vertices of the {\sf LMC} with respect to (a) the number of vertices $n$ (c) the proportion of happiness $\rho$ and (e) the number of colours $k$. The same charts for improvements by \hyperref[alg:ls]{\sf LS} on the outputs of {\sf Growth} are demonstrated in figures (b), (d) and (f). 
    }\label{fig:imp-lg}
    
\end{figure}

When $\rho$ is very small, a solution produced by {\sf LMC} usually has a high ratio of $\rho$-happy vertices, hence, there is little room for improvement. In contrast, when $\rho$ approaches 1, {\sf LMC} generates many fewer $\rho$-happy vertices, and therefore the improvement algorithms can enhance the quality of the solution output by {\sf LMC} (Figure~\ref{fig:LMC_r_improved}). On the other hand, when $\rho$ is very small, \hyperref[alg:ls]{\sf LS} and \hyperref[alg:rls]{\sf RLS} dramatically improve upon the solution generated by {\sf Growth}, making it very close to a complete $\rho$-happy colouring (Figure~\ref{fig:Growth_r_improved}).  While \hyperref[alg:rls]{\sf RLS} shows a smooth trend that is not contingent on the trend of the inputs, \hyperref[alg:ls]{\sf LS} improvement fluctuates along with that of {\sf Growth}, when $\rho\leq 0.8$. When $\rho$ approaches 1, none of the improvement algorithms can enhance the solutions generated by {\sf Growth}. Furthermore, when $\rho\geq 0.5$, both \hyperref[alg:ls]{\sf LS} and \hyperref[alg:rls]{\sf RLS} show similar levels of improvement over the solutions produced by {\sf Growth}, but when $\rho\leq 0.5$, the \hyperref[alg:rls]{\sf RLS} improvement is almost always better than that of \hyperref[alg:ls]{\sf LS}.

In our tests, the number of colours is $2\leq k\leq 20$. By Figures~\ref{fig:LMC_k_improved} and \ref{fig:Growth_k_improved}, it can be seen that an increase in $k$ results in a decrease of the ratio of $\rho$-happy vertices in the solutions generated by both {\sf LMC} and {\sf Growth}. Thus, when $k$ is larger, there is a larger room for improvements of the ratio. Interestingly, this increase in the number of colours causes almost similar fluctuations in the amounts of improvements by the three improvement algorithms, which align with the input colourings generated by {\sf LMC} and {\sf Growth}.

Figure~\ref{fig:imp} provides greater insight into the improvements that can be achieved by our three improvement algorithms. While the focus of Figure~\ref{fig:imp-lg} is on the average values of the ratio of $\rho$-happy vertices, it focuses on the ranges of improvements (the number of improved vertices divided by the total number of vertices $n$). We observe that the minimum improvement by \hyperref[alg:ls]{\sf LS} is usually close to 0, but the maximum improvement can equal the total number of vertices. This implies that there are cases where {\sf LMC} or {\sf Growth} generate colourings with almost no $\rho$-happy vertex, but by integrating with an improvement algorithm, they are able to achieve solutions with complete $\rho$-happy colourings. This interesting property is fundamental to devise search-based and population-based metaheuristics. 

\begin{figure}
\captionsetup{size=small}
\begin{subfigure}{0.48\textwidth}
    \includegraphics[scale=0.48]{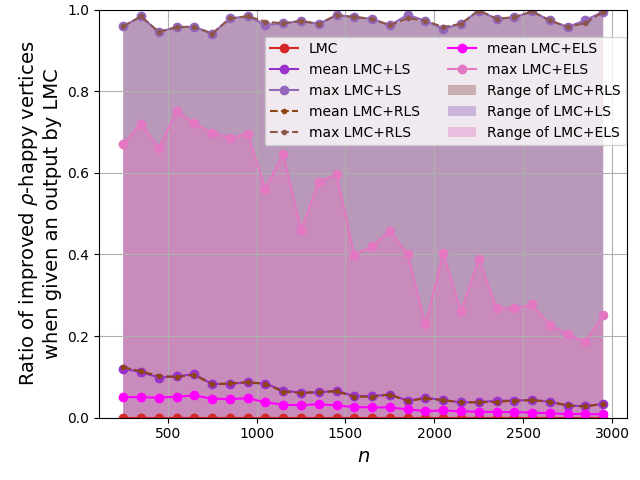}
    \caption{}\label{fig:LMC_n_imp}
\end{subfigure} 
\hspace{5mm}
\begin{subfigure}{0.48\textwidth}
    \includegraphics[scale=0.48]{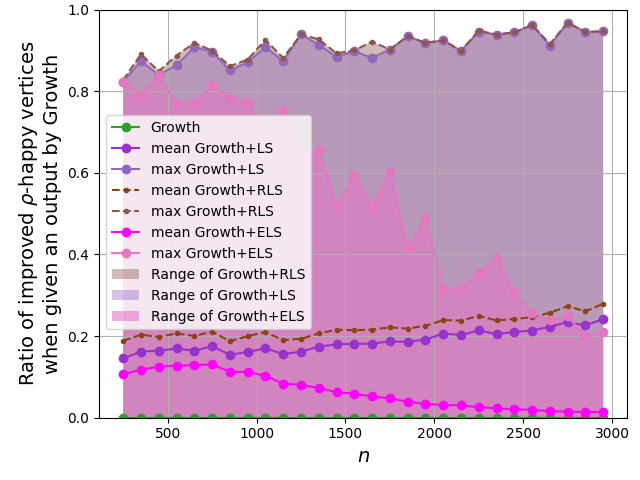}
    \caption{}\label{fig:Growth_n_imp}
\end{subfigure} 

\begin{subfigure}{0.48\textwidth}  
    \includegraphics[scale=0.48]{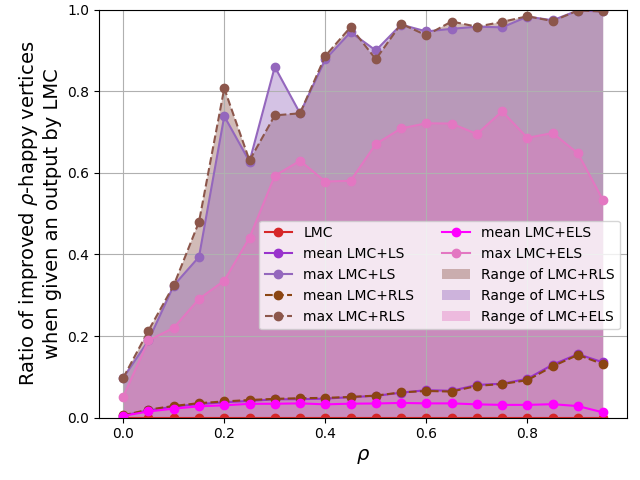}
    \caption{}\label{fig:LMC_r_imp_}
\end{subfigure} 
\hspace{5mm}
\begin{subfigure}{0.48\textwidth}  
    \includegraphics[scale=0.48]{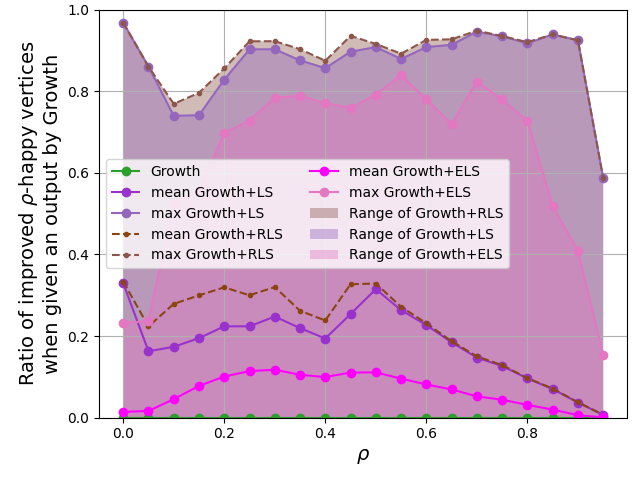}
    \caption{}\label{fig:Growth_r_imp}
\end{subfigure}

\begin{subfigure}{0.48\textwidth}  
    \includegraphics[scale=0.48]{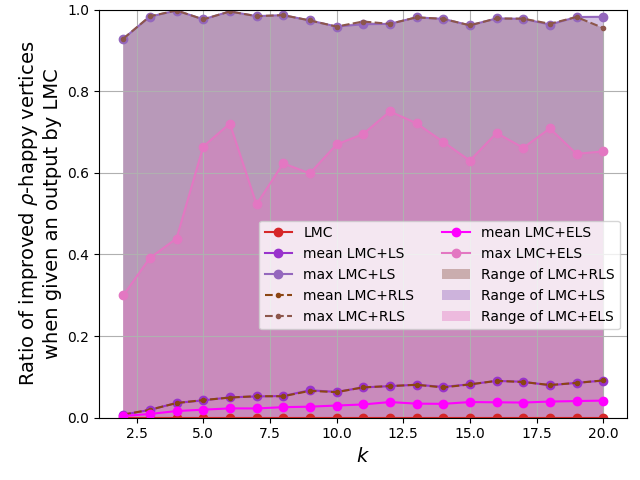}
    \caption{}\label{fig:LMC_k_imp}
\end{subfigure} 
\hspace{5mm}
\begin{subfigure}{0.48\textwidth}  
    \includegraphics[scale=0.48]{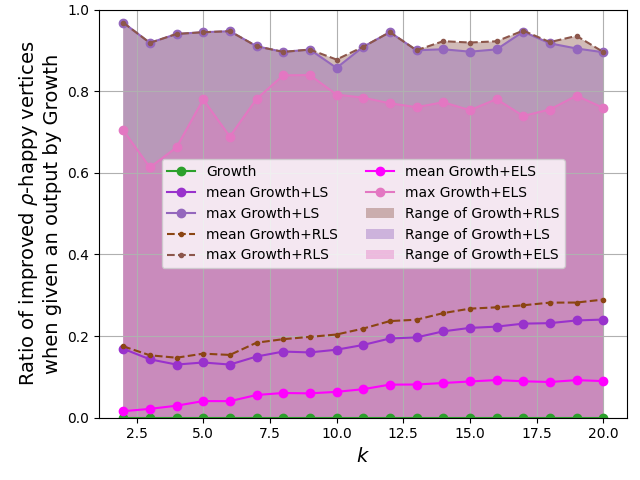}
    \caption{}\label{fig:Growth_k_imp}
\end{subfigure} 
    
    \caption{Minimum, mean and maximum enhancements by the improvement algorithms (\hyperref[alg:ls]{\sf LS}, \hyperref[alg:rls]{\sf RLS}, and \hyperref[alg:els]{\sf ELS}) on the ratio of $\rho$-happy vertices over the output of {\sf LMC} with respect to (a) $n$, (c) $\rho$, and (e) $k$. The same charts for enhancing the outputs of {\sf Growth} are shown in Figures (b), (d) and (f), respectively.
    }\label{fig:imp}
    
\end{figure}

All charts in Figure~\ref{fig:imp} show that \hyperref[alg:els]{\sf ELS} (pink) always falls behind \hyperref[alg:ls]{\sf LS} (purple) and \hyperref[alg:rls]{\sf RLS} (brown), for the average and maximum amount of improvements. For improving {\sf LMC}, both \hyperref[alg:ls]{\sf LS} and \hyperref[alg:rls]{\sf RLS} perform at a similar level considering their average and maximum values (Figures~\ref{fig:LMC_n_imp}, \ref{fig:LMC_r_imp_}, and \ref{fig:LMC_k_imp}). For their maximum values, there are only a few occasional differences, which are clearly visible in Figure~\ref{fig:LMC_r_imp_} when $0.1\leq \rho\leq0.4$. 

Regarding improving the solutions produced by {\sf Growth}, it is clearly visible that, on average, \hyperref[alg:rls]{\sf RLS} finds more $\rho$-happy vertices than \hyperref[alg:ls]{\sf LS}, and occasionally, even its maximum number of improved $\rho$-happy vertices is higher that of \hyperref[alg:ls]{\sf LS} (Figures~\ref{fig:Growth_n_imp}, \ref{fig:Growth_r_imp}, and \ref{fig:Growth_k_imp}). However, these maximum values are very close and sometimes even indistinguishable, for example, when $n\geq 2000$ (Figure~\ref{fig:Growth_n_imp}).

\begin{table}[H]
\centering
\begin{NiceTabular}{*{4}{c}*{4}{c}}[hvlines, code-before =
\rectanglecolor{vibtilg!50}{1-1}{1-8}
\rectanglecolor{C_LMC!40}{2-1}{4-4}
\rectanglecolor{C_LS!50}{2-5}{2-8}
\rectanglecolor{C_ELS!35}{3-5}{3-8}
\rectanglecolor{C_PLS!50}{4-5}{4-8}
\rectanglecolor{C_growth!40}{5-1}{7-4}
\rectanglecolor{C_LS!50}{5-5}{5-8}
\rectanglecolor{C_ELS!35}{6-5}{6-8}
\rectanglecolor{C_PLS!50}{7-5}{7-8}
\rectanglecolor{C_rnd!40}{8-1}{10-4}
\rectanglecolor{C_LS!50}{8-5}{8-8}
\rectanglecolor{C_ELS!35}{9-5}{9-8}
\rectanglecolor{C_PLS!50}{10-5}{10-8}
\rectanglecolor{C_LS!50}{11-1}{12-4}
\rectanglecolor{C_ELS!35}{11-5}{11-8}
\rectanglecolor{C_PLS!50}{12-5}{12-8}
\rectanglecolor{C_ELS!35}{13-1}{13-4}
\rectanglecolor{C_PLS!50}{13-5}{13-8}
\rectanglecolor{C_PLS!50}{14-1}{14-4}
\rectanglecolor{black!60}{14-5}{14-8}]
{\bf Base} & $\boldsymbol{\alpha(\sigma)}^\ast$ & {\bf ACD$^\dagger$} & {\bf Time (s)} & {\bf Improvement} &$\boldsymbol{\alpha(\tilde{\sigma})}^\ast$ & {\bf ACD$^\dagger$} & {\bf Total Time (s)}   \\ \hline
\Block{3-1}{\sf LMC} & \Block{3-1}{0.7187} & \Block{3-1}{0.4618} & \Block{3-1}{0.798} & 
{\sf LS} & 0.7825  &  0.4648  & 1.845  \\ 
& & & &{\sf RLS} & 0.7823 &  0.4648  & 2.451 \\ 
& & & & {\sf ELS} & 0.7476 & 0.4631  & 38.01  \\ \hline
\Block{3-1}{\sf Growth} & \Block{3-1}{0.6843} & \Block{3-1}{0.216} & \Block{3-1}{30.831}&  
{\sf LS} & 0.8692  & 0.2476  & 31.856 \\ 
& & & & {\sf RLS} &  0.9044  & 0.2571  & 33.407 \\ 
& & & & {\sf ELS} & 0.7497  & 0.2228  & 79.137  \\ \hline
\Block{3-1}{\sf Random}& \Block{3-1}{0.1376} & \Block{3-1}{0.19} &\Block{3-1}{1.421}& {\sf LS} & 0.8088 & 0.2935 & 2.806\\
& & & & {\sf RLS} & 0.6037 & 0.2806 & 3.876 \\
& & & & {\sf ELS} & 0.1649 & 0.195 & 57.095 \\ \hline
\Block{2-1}{\sf LS} & \Block{2-1}{0.8771} & \Block{2-1}{0.2325} & \Block{2-1}{1.009}&  
{\sf RLS} & 0.8929 & 0.2345  & 2.138  \\ 
& & & & {\sf ELS} & 0.8843 & 0.2337  & 14.275  \\ \hline
{\sf RLS} & 0.5377& 0.1948 & 1.934 & {\sf ELS} & 0.5378  & 0.1948  & 29.67  \\ \hline
{\sf ELS} & 0.0216 & 0.066 & 59.845 & & & & \\ \hline

\end{NiceTabular}
\medskip
\captionsetup{singlelinecheck=off}
\caption[foo bar]{Average values of ratio of $\rho$-happy vertices, accuracy of community detection, and running times of the heuristic algorithms {\sf LMC}, {\sf Growth}, {\sf Random}, \hyperref[alg:ls]{\sf LS}, \hyperref[alg:rls]{\sf RLS}, and \hyperref[alg:els]{\sf ELS}, followed by improvement algorithms \hyperref[alg:ls]{\sf LS}, \hyperref[alg:rls]{\sf RLS}, and \hyperref[alg:els]{\sf ELS}. Each row shows how much an improvement algorithm enhanced the ratio of $\rho$-happy vertices and the accuracy of community detection, and how much time consumed for such improvement. 
\begin{itemize}
\item[$\ast$.] $\alpha(\sigma)=\frac{H_\rho (\sigma)}{n}$, the ratio of $\rho$-happy vertices.
\item[$\dagger$.] Accuracy of community detection, see Equation~\ref{eq:acd}.
\end{itemize}}
\label{table:1}
\end{table}

Table~\ref{table:1} contains the averages of the ratio of $\rho$-happy vertices $\alpha(\sigma)$, the accuracy of community detection $ACD(\sigma)$, and the running time of the heuristics and improvements. It shows these average values for the solutions obtained by {\sf LMC}, {\sf Growth}, {\sf Random} (randomly assigns colours to the vertices), \hyperref[alg:ls]{\sf LS}, \hyperref[alg:rls]{\sf RLS}, and \hyperref[alg:els]{\sf ELS} in the four columns on the left-hand side of the table, and the effects of the improvement algorithms \hyperref[alg:ls]{\sf LS}, \hyperref[alg:rls]{\sf RLS}, and \hyperref[alg:els]{\sf ELS} on them. Figures~\ref{fig:grouped-imp-hap}, \ref{fig:grouped-imp-comm}, and \ref{fig:grouped-imp-time} help to analyse data from Table~\ref{table:1}.

\begin{figure}
    \centering
    \includegraphics[scale=0.5]{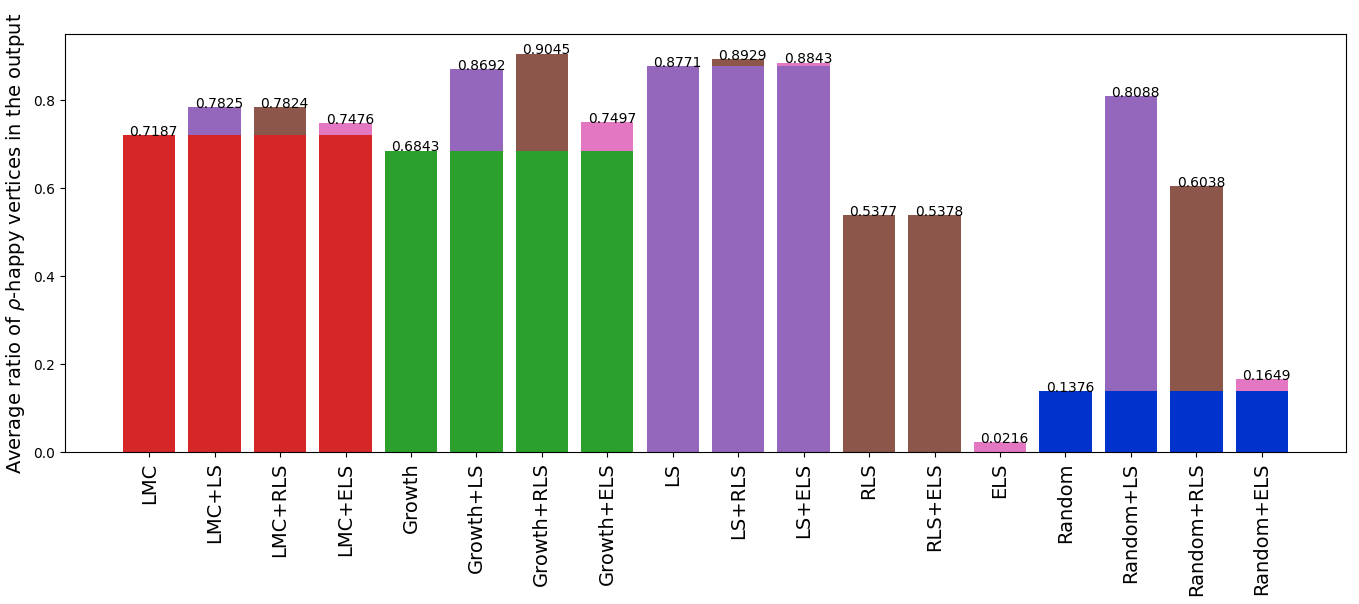}
    \caption{This chart shows how much each of the improvement algorithms, i.e., \hyperref[alg:ls]{\sf LS}, \hyperref[alg:rls]{\sf RLS} and \hyperref[alg:els]{\sf ELS}, improves the ratio of $\rho$-happy vertices of the outputs generated by the heuristics, namely {\sf LMC}, {\sf Growth}, {\sf LS} and \hyperref[alg:rls]{\sf RLS}.}
    \label{fig:grouped-imp-hap}
\end{figure}

The highest average $\rho$-happy colouring ratio is for {\sf Growth+ELS} with around \% 90. This is caused by the massive \% 32 improvement caused by \hyperref[alg:rls]{\sf RLS} over the not impressive average result \% 68 of {\sf Growth}. However, this high result comes with a high time cost of an average of 33.4 seconds, 30.8 seconds of which are used by {\sf Growth} itself. The second best average of the ratio of $\rho$-happy vertices is {\sf LS+ELS} with more than \% 89. The time cost is quite low, only 2.138 s out of 120.0 s of the total time limit. These are visible in Figure~\ref{fig:grouped-imp-hap}, which illustrates the average ratios of $\rho$-happiness of the input solutions and the averages of the improvements made by the improvement algorithms.

\begin{figure}
    \centering
    \includegraphics[scale=0.48]{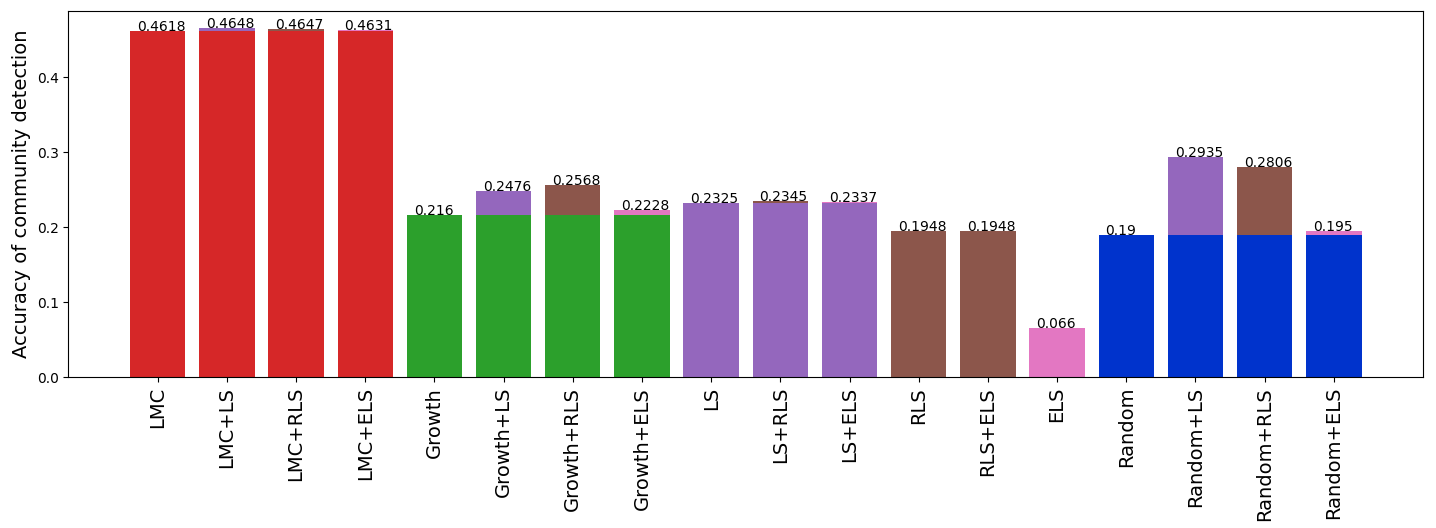}
    \caption{This chart shows how each of the improvement algorithms, i.e., \hyperref[alg:ls]{\sf LS}, \hyperref[alg:rls]{\sf RLS} and \hyperref[alg:els]{\sf ELS}, improve the accuracy of community detection of the heuristics, namely {\sf LMC}, {\sf Growth}, {\sf LS} and \hyperref[alg:rls]{\sf RLS}.}
    \label{fig:grouped-imp-comm}
\end{figure}

Evidenced by the data in Table~\ref{table:1}, none of our improvement algorithms negatively impacts the average accuracy of community detection (i.e., the accuracy of community detection is always as good as the original solution). Among these, the best average community detection is obtained from running {\sf LMC} with a very small enhancement by an improvement algorithm. Figure~\ref{fig:grouped-imp-comm} demonstrates how the three improvement algorithms affect the average community detection of the solutions given by the heuristics. Based on it, the best result by far for the average accuracy of community detection belongs to algorithms start with a solution generated by {\sf LMC}.

\begin{figure}
    \centering
    \includegraphics[scale=0.5]{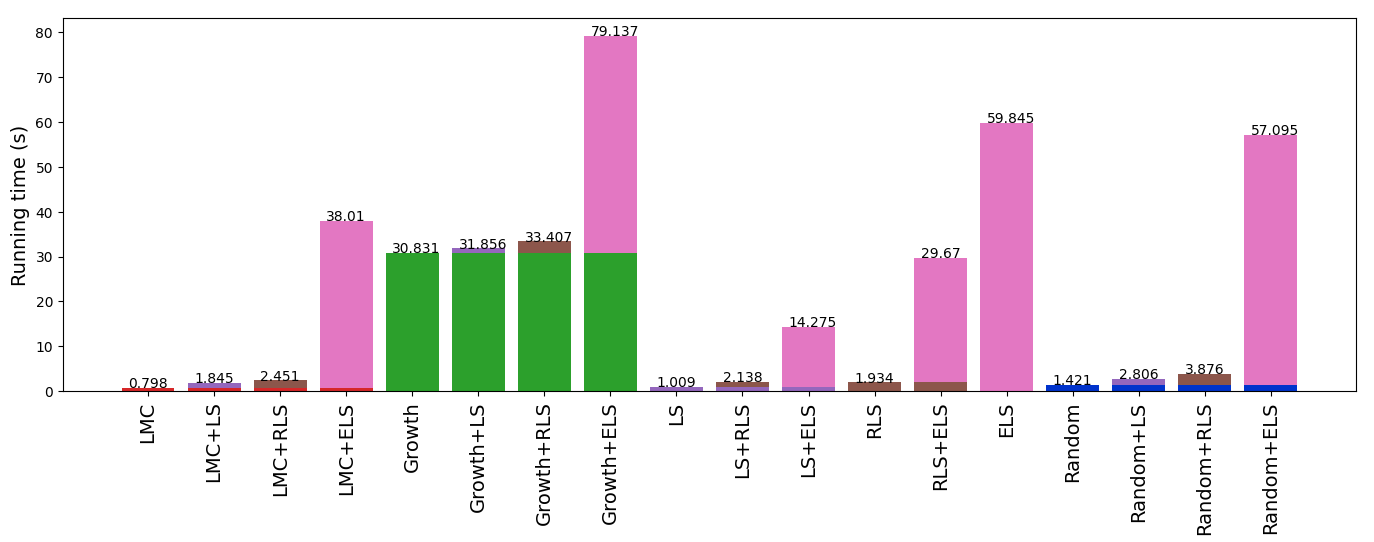}
    \caption{This chart shows the running time of the improvement algorithms, i.e., \hyperref[alg:ls]{\sf LS}, \hyperref[alg:rls]{\sf RLS} and \hyperref[alg:els]{\sf ELS}, trying to improve the number of $\rho$-happy vertices in the output of the heuristics, namely {\sf LMC}, {\sf Growth}, {\sf LS} and \hyperref[alg:rls]{\sf RLS}.}
    \label{fig:grouped-imp-time}
\end{figure}

 The best average running time belongs to {\sf LMC+LS} with the total average of 1.845 s out of 120 s. The second best is {\sf LS+ELS} with 2.138 s with an impressive \% 89 of the average ratio of $\rho$-happy vertices. The best improvement algorithm with respect to its running time --- among the three improvement algorithms we consider here --- is \hyperref[alg:ls]{\sf LS}, with an average of around 1.0 s. The average time of \hyperref[alg:rls]{\sf RLS} is more than twice that of \hyperref[alg:ls]{\sf LS}. However, the average running time of \hyperref[alg:els]{\sf ELS} is almost always near the time limit, indicating that its outputs must occasionally be reported from an incomplete search. Figure~\ref{fig:grouped-imp-time} presents an illustration of the average running time of the improvement algorithms over the solutions given by heuristics.

\subsection{Investigating the effects of the number of precoloured vertices}\label{sec:pcc}

In this section, we consider the effects of increasing the number of precoloured vertices per community ($pcc$). Visiting our previous analyses (Section~\ref{sec:details}) it seemed interesting to see the effects of such an increase in this parameter. Therefore, we conduct experiments on a new set of graphs where the $pcc$ is set to be $1\leq pcc\leq \left\lfloor \frac{n}{k}\right\rfloor$ and maintaining all other parameter settings.

\begin{figure}
\captionsetup{size=small}
\begin{subfigure}{0.48\textwidth}
    \includegraphics[scale=0.5]{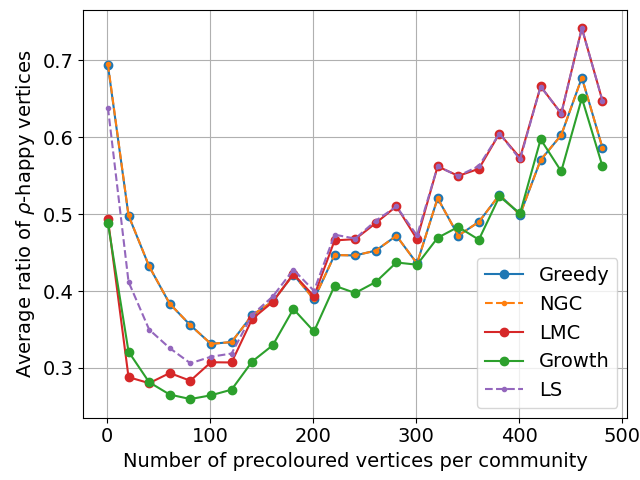}
    \caption{}\label{fig:pcc-happy}
\end{subfigure} 
\hspace{5mm}
\begin{subfigure}{0.48\textwidth}
    \includegraphics[scale=0.5]{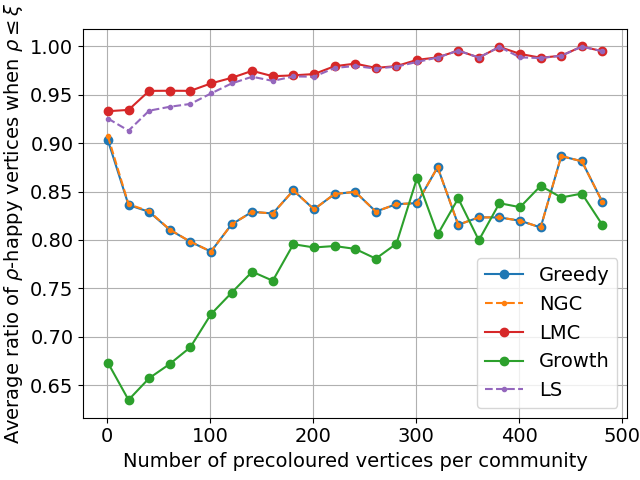}
    \caption{}\label{fig:pcc-happy-xi}
\end{subfigure} 

\begin{subfigure}{0.48\textwidth}
    \includegraphics[scale=0.5]{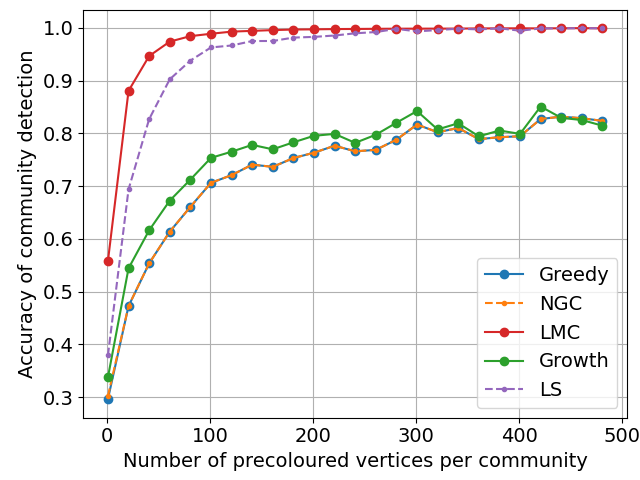}
    \caption{}\label{fig:pcc-comm}
\end{subfigure} 
\hspace{5mm}
\begin{subfigure}{0.48\textwidth}
    \includegraphics[scale=0.5]{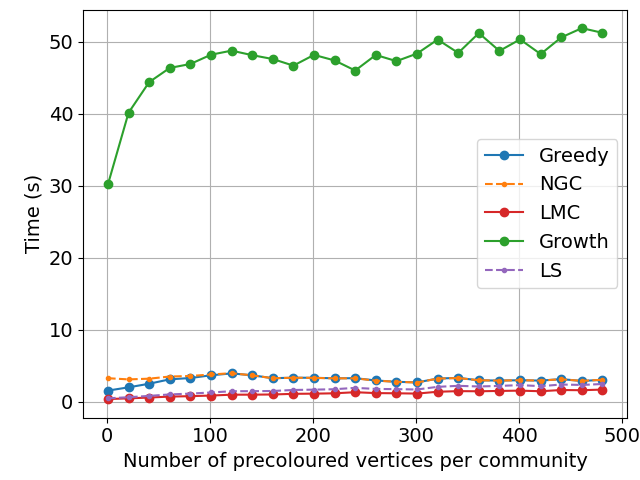}
    \caption{}\label{fig:pcc-time}
\end{subfigure} 
    \caption{The effects of increasing the number of precoloured vertices per community; (a) average ratio of $\rho$-happy vertices in general, 
    (b) average ratio of $\rho$-happy vertices when $\rho\leq \xi$, (c) average accuracy of community detection, and (d) algorithms' running time. }
    \label{fig:pcc}
\end{figure}

The results of running the algorithms {\sf LMC} (red), {\sf Growth} (green), {\sf Greedy} (blue), {\sf NGC} (orange, dashed), and \hyperref[alg:ls]{\sf LS} (purple, dashed) are presented in Figure~\ref{fig:pcc}. It can be seen in Figure~\ref{fig:pcc-happy} that when $pcc$ is very low, all heuristics generate a high ratio of happy vertices, with {\sf Greedy} and {\sf NGC} performing best. As $pcc$ approaches 100, these ratios drop dramatically, enabling {\sf LMC} to outperform the other algorithms. In the range $150\leq pcc \leq 500$, {\sf LMC} performs best while {\sf Greedy} and {\sf NGC} are unable to even compete with {\sf Growth}. In general, when $pcc > 120$, the ratios of $\rho$-happy vertices show an increase pattern until $pcc$ reaches 500. 

When $\rho\leq \xi$, the colouring induced by the communities can make every vertex $\rho$-happy~\cite{SHEKARRIZ2025106893}, hence, increasing $pcc$ results in higher ratios of $\rho$-happy colouring in the solutions generated by {\sf LMC}, {\sf Growth} and \hyperref[alg:ls]{\sf LS}, but less so for {\sf Greedy} and {\sf NGC} (Figure~\ref{fig:pcc-happy-xi}). While {\sf LMC} and \hyperref[alg:ls]{\sf LS} perform best, their difference becomes negligible as $pcc$ increases.

Increasing $pcc$ has an obvious effect on the accuracy of community detection (Figure~\ref{fig:pcc-comm}). Although {\sf LMC} and \hyperref[alg:ls]{\sf LS} detect communities with high precision when $pcc$ increases, {\sf Greedy}, {\sf NGC} and {\sf Growth} cannot reach an accuracy of more than 85\%, although their community detection precision also improves with an increase in $pcc$.

The running time of the tested algorithms experiences some slight effects with $pcc$ being increased; see Figure~\ref{fig:pcc-time}. The average running times of {\sf LMC} and \hyperref[alg:ls]{\sf LS} increase very slightly with an increase of $pcc$, but hardly reach 4 s as $pcc$ approaches 500. {\sf Greedy} and {\sf NGC}'s running times remain almost neutral to $pcc$'s change, while {\sf Growth}'s average running time increases sharply from 30 s to 50 s when $pcc$ increases from 1 to 120, and then fluctuates around 50 s as $pcc$ further increases. 

To sum-up, the reader must be reminded that the good performance of {\sf Greedy} and {\sf NGC} in the former tests is due to a large number of free vertices. Assigning only one colour to all free vertices generates a high number of $\rho$-happy vertices when the majority of the vertices are free vertices. Consequently, it is not worth considering {\sf Greedy} and {\sf NGC} to design metaheuristics and mathematical strategies for soft happy colouring when the number of precoloured vertices can be high. Another reason for this is that these algorithms have no randomness employed to generate a solution, so running them several times does not make any difference. 
Moreover, their outputs have nothing to do with the graphs' community structure, the main motivation of defining and studying (soft) happy colouring. Additionally, {\sf LMC} and \hyperref[alg:ls]{\sf LS} have lower time complexity, higher degrees of correlation with graph's communities, and higher ratios of $\rho$-happy vertices (especially when $\rho\leq \xi$), whose results can be different each time they are run.

\section{Conclusion}\label{sec:conc} 

This paper investigates soft happy colouring for graphs in the SBM, contributing to both theory and practice. One of the theoretical results shows that the accuracy of community detection of a complete $\rho$-happy colouring is expected to increase when $\rho$ becomes larger. Another result is that an algorithm can almost never find a complete happy coloring $\rho$ when the number of vertices is large and $\rho$ is greater than the threshold $$\tilde{\xi}=\frac{p}{p+(k-1)q}.$$ These theoretical results are experimentally verified for a large set of randomly generated partially coloured graphs in the SBM. Hence, practically, algorithms designed to detect communities of a graph modelled by the SBM should focus on $\rho$ being slightly less than $\tilde{\xi}$. 

We introduced three new algorithms, \hyperref[alg:ls]{\sf LS}, \hyperref[alg:rls]{\sf RLS}, and \hyperref[alg:els]{\sf ELS}, that can alter a (probably partial) solution and generate a higher quality solution. We tested these and other existing algorithms on a large set of randomly generated partially coloured graphs in the SBM. The results show that \hyperref[alg:ls]{\sf LS} has low time requirements and is effective as a heuristic or an improvement algorithm, especially when $\rho\leq \tilde{\xi}$ or given a random input. Moreover, \hyperref[alg:rls]{\sf RLS} demonstrates improved performance only when it is provided a solution of {\sf Growth}'s as an input. The algorithm \hyperref[alg:els]{\sf ELS} showed poor performance due to its high time complexity. 

In practice, \hyperref[alg:ls]{\sf LS} performs comparatively well considering the other heuristics. It not only generates high number of $\rho$-happy vertices, but is also very fast. The test results, along with short run-time requirements of \hyperref[alg:ls]{\sf LS} mean that it is a good candidate to be used in search-based and population-based metaheuristic algorithms. It can be expected that the ensuing results can be of similar or even better quality solutions, owing to exploring substantially larger areas of the search space. 

Since finding a fast and effective heuristic is essential for designing metaheuristics and matheuristics, such heuristics can be used to generate a population of solutions for population-based metaheuristics or to iterate improvements in perturbation-based algorithms. Therefore, we propose using \hyperref[alg:ls]{\sf LS} to design more advanced algorithms for soft happy colouring.

\section*{Declarations of interest} 
None

\bibliographystyle{plain}
\bibliography{HC.bib}

\begin{thebibliography}{10}

\bibitem{Bollobas_2001}
Béla Bollobás.
\newblock {\em Random Graphs}.
\newblock Cambridge Studies in Advanced Mathematics. Cambridge University
  Press, 2nd edition, 2001.

\bibitem{Carpentier2023}
Louis Carpentier, Jorik Jooken, and Jan Goedgebeur.
\newblock A heuristic algorithm using tree decompositions for the maximum happy
  vertices problem.
\newblock {\em Journal of Heuristics}, Nov 2023.

\bibitem{Cherifi2019}
Hocine Cherifi, Gergely Palla, Boleslaw~K. Szymanski, and Xiaoyan Lu.
\newblock On community structure in complex networks: challenges and
  opportunities.
\newblock {\em Applied Network Science}, 4(1):117, Dec 2019.

\bibitem{10.1007/978-3-540-48413-4_23}
Anne Condon and Richard~M. Karp.
\newblock Algorithms for graph partitioning on the planted partition model.
\newblock In Dorit~S. Hochbaum, Klaus Jansen, Jos{\'e} D.~P. Rolim, and
  Alistair Sinclair, editors, {\em Randomization, Approximation, and
  Combinatorial Optimization. Algorithms and Techniques}, pages 221--232,
  Berlin, Heidelberg, 1999. Springer Berlin Heidelberg.

\bibitem{doi:10.1073/pnas.122653799}
M.~Girvan and M.~E.~J. Newman.
\newblock Community structure in social and biological networks.
\newblock {\em Proceedings of the National Academy of Sciences},
  99(12):7821--7826, 2002.

\bibitem{PhysRevE.68.065103}
R.~Guimer\`a, L.~Danon, A.~D\'{\i}az-Guilera, F.~Giralt, and A.~Arenas.
\newblock Self-similar community structure in a network of human interactions.
\newblock {\em Phys. Rev. E}, 68:065103, Dec 2003.

\bibitem{HOLLAND1983109}
Paul~W. Holland, Kathryn~Blackmond Laskey, and Samuel Leinhardt.
\newblock Stochastic blockmodels: First steps.
\newblock {\em Social Networks}, 5(2):109--137, 1983.

\bibitem{JERRUM1998155}
Mark Jerrum and Gregory~B. Sorkin.
\newblock The metropolis algorithm for graph bisection.
\newblock {\em Discrete Applied Mathematics}, 82(1):155--175, 1998.

\bibitem{Lee2019}
Clement Lee and Darren~J. Wilkinson.
\newblock A review of stochastic block models and extensions for graph
  clustering.
\newblock {\em Applied Network Science}, 4(1):122, Dec 2019.

\bibitem{Lewis2019265}
Rhyd Lewis, Dhananjay Thiruvady, and Kerri Morgan.
\newblock Finding happiness: An analysis of the maximum happy vertices problem.
\newblock {\em Computers \& Operations Research}, 103:265--276, 2019.

\bibitem{LEWIS2021105114}
Rhyd Lewis, Dhananjay Thiruvady, and Kerri Morgan.
\newblock The maximum happy induced subgraph problem: Bounds and algorithms.
\newblock {\em Computers \& Operations Research}, 126:105114, 2021.

\bibitem{lin2011probability}
Z.~Lin and Z.~Bai.
\newblock {\em Probability Inequalities}.
\newblock Springer Berlin Heidelberg, 2011.

\bibitem{Homophily}
Miller McPherson, Lynn Smith-Lovin, and James~M. Cook.
\newblock Birds of a feather: Homophily in social networks.
\newblock {\em Annual Review of Sociology}, 27(1):415--444, 2001.

\bibitem{SHEKARRIZ2025106893}
Mohammad~H. Shekarriz, Dhananjay Thiruvady, Asef Nazari, and Rhyd Lewis.
\newblock Soft happy colourings and community structure of networks.
\newblock {\em Computers \& Operations Research}, 174:106893, 2025.

\bibitem{THIRUVADY2022101188}
Dhananjay Thiruvady and Rhyd Lewis.
\newblock Recombinative approaches for the maximum happy vertices problem.
\newblock {\em Swarm and Evolutionary Computation}, 75:101188, 2022.

\bibitem{thiruvady2020}
Dhananjay Thiruvady, Rhyd Lewis, and Kerri Morgan.
\newblock {Tackling the Maximum Happy Vertices Problem in Large Networks}.
\newblock {\em 4OR}, 18(4):507--527, 2020.

\bibitem{ZHANG2015117}
Peng Zhang and Angsheng Li.
\newblock Algorithmic aspects of homophyly of networks.
\newblock {\em Theoretical Computer Science}, 593:117--131, 2015.

\bibitem{Zhang2018}
Peng Zhang, Yao Xu, Tao Jiang, Angsheng Li, Guohui Lin, and Eiji Miyano.
\newblock Improved approximation algorithms for the maximum happy vertices and
  edges problems.
\newblock {\em Algorithmica}, 80(5):1412--1438, May 2018.

\end{thebibliography}
\end{document}